\documentclass[12pt,onecolumn]{IEEEtran} %
\usepackage{cite}
\usepackage{amsmath,amssymb,amsfonts}
\usepackage{algorithmic}
\usepackage{graphicx}
\usepackage{textcomp}
\usepackage{xcolor}

\usepackage{bm}
\usepackage{bbm}
\usepackage{subfigure}
\usepackage{color, colortbl}
\usepackage{soul}
\usepackage{url}
\usepackage{enumitem} 
\usepackage[ruled,vlined]{algorithm2e} 
\usepackage{amsmath}
\usepackage{amsthm}
\usepackage{tikz} 
\usepackage{pgfplots}
\usepackage[normalem]{ulem}
\pgfplotsset{compat=1.5}

\usepackage{filecontents}

\usepackage{floatflt}
\usepackage{graphics}

\usetikzlibrary{arrows}

\newtheorem{theorem}{Theorem}
\newtheorem{lemma}{Lemma}
\newtheorem{corollary}{Corollary}
\newtheorem{definition}{Definition}

\newcommand{\allClasses}{{\mathcal{N}}}
\newcommand{\leafClasses}{{\mathcal{L}}}
\newcommand{\internalClasses}{{\mathcal{I}}}
\newcommand{\rootClass}{{\rm root}}

\newcommand{\parent}[1]{{{\rm p}(#1)}}
\newcommand{\child}[1]{{{\rm child}(#1)}}
\newcommand{\anc}[1]{{{\rm anc}(#1)}}
\newcommand{\sib}[1]{{{\rm sib}(#1)}}
\newcommand{\desc}[1]{{{\rm desc}(#1)}}
\newcommand{\ldesc}[1]{{{\rm ldesc}(#1)}}
\newcommand{\stree}[1]{{{\rm stree}(#1)}}

\newcommand{\lmax}[1]{L_{#1}^{\max}}

\newcommand{\weight}[2][]{{w^{#1}_{#2}}}
\newcommand{\fairQuota}[2][]{{F^{#1}_{#2}}}
\newcommand{\balance}[2][]{{B^{#1}_{#2}}}
\newcommand{\aggBalance}[2][]{{A^{#1}_{#2}}}
\newcommand{\residual}[2][]{{R^{#1}_{#2}}}

\newcommand{\rc}[1]{{\rm rc}(#1)}
\newcommand{\activeClasses}{{\mathcal{A}}}

\newcommand{\reffig}[1]{Fig.~\ref{#1}}

\newcommand{\refeq}[1]{\eqref{#1}}
\newcommand{\reflem}[1]{Lemma~\ref{#1}}

\def\N{\mathcal{N}}
\def\S{\mathcal{S}}

\def\I{\mathcal{I}}
\def\L{\mathcal{L}}
\def\A{\mathcal{A}}

\begin{document}
\title{A Round-Robin Packet Scheduler for Hierarchical Max-Min Fairness}

\author{
Natchanon Luangsomboon, ~J\"{o}rg Liebeherr
\vspace{-4mm}
\thanks{
N. Luangsomboon and J. Liebeherr  are with the Department of Electrical and Computer Engineering, University of Toronto. 
}
\thanks{
The  research in this paper is supported in part by an NSERC Discovery grants.}
}
\maketitle

\begin{abstract}
Hierarchical link sharing addresses the demand for fine-grain traffic control at multiple levels of aggregation. At present, packet schedulers  that can support hierarchical link sharing 
are not suitable for an implementation at line rates, and deployed schedulers perform poorly when distributing excess capacity to classes that need additional bandwidth. We present {\it HLS}, a packet scheduler that ensures a hierarchical max-min fair allocation of the link bandwidth. HLS supports minimum rate guarantees and isolation between classes. Since it is realized as a non-hierarchical round-robin scheduler, it is suitable to operate at high rates.   
We implement HLS in the Linux kernel and evaluate it with respect to achieved 
rate allocations and overhead. We compare the results with those obtained for
CBQ and HTB, the existing scheduling algorithms in Linux for hierarchical link sharing. We show that the overhead of HLS is comparable to that of other classful packet schedulers. 
\end{abstract}

\section{Introduction}
\label{sec:intro}
Packet scheduling plays a crucial 
role in  the management of traffic flows, for prioritizing traffic,  for flexible service differentiation, and for achieving performance metrics, such as flow completion times, throughput, and the tail of the delay distribution. 
This paper is concerned with packet scheduling methods that support traffic control at multiple aggregation levels. The need for such scheduling methods is  largely driven by content providers that  manage traffic within and between servers, clusters, and data centers. Increasingly, data centers rely on fine-grain traffic control at multiple levels of aggregation.  The Google~B4 inter data center network reports no less than five levels of traffic aggregation \cite{B4,datacenter-tc}. 
Traffic control in support of a hierarchical distribution of available bandwidth is referred to as {\em hierarchical link sharing}. 

As an example of link sharing, consider the hierarchy shown in Fig.~\ref{fig:hierarchy}. 
The top of the hierarchy,  labeled as {\em root}, is a link with a 
fixed rate of 1000 (units are in Mbps). This bandwidth is to be divided between three traffic classes $A$,~$B$, and~$C$ that each receive a minimum rate guarantee,  
as indicated in the figure. 
Traffic class~$A$ is further divided into classes $A1$ and $A2$, with guarantees of $100$ and $200$, respectively. Class $B$ splits the bandwidth between 
$B1$ and $B2$ in the same fashion. 
Arriving packets are classified and mapped to {\it leaf classes}, which are 
the classes at the bottom  level of the hierarchy. 

\begin{figure}

  \centering
 	\includegraphics[width=0.5\columnwidth]{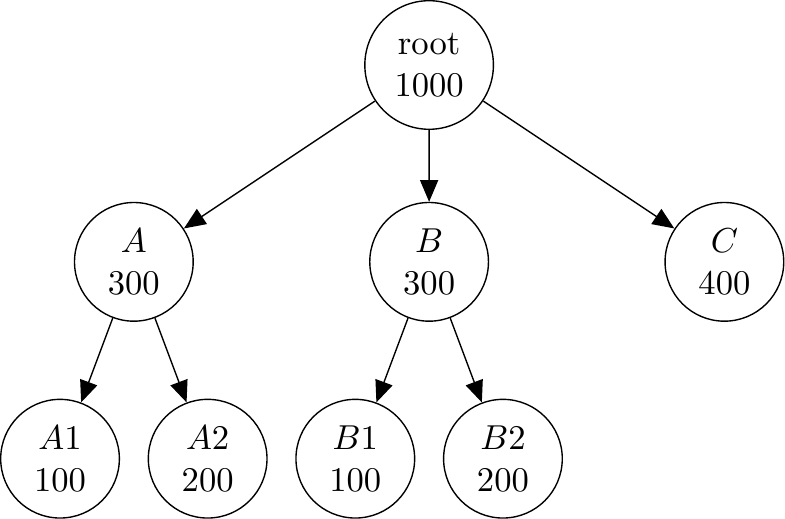}

\caption{Link sharing hierarchy.}
\label{fig:hierarchy}

\vspace{-7mm}
\end{figure}

Clearly, if the aggregate traffic from all leaf classes does not exceed the 
link capacity,  every leaf class can obtain a rate equal to its arrival rate. Likewise, if the arrival rate of every leaf class exceeds its guaranteed rate, then each leaf class is limited to its  
guaranteed rate.  
The bandwidth allocation becomes less trivial when the aggregate arrival rate from all 
classes is larger than the link 
capacity, and some classes exceed 
their guaranteed rates, while others stay well below their guarantees. 
In this case, excess capacity left unused by some classes must be distributed equitably to classes that desire additional bandwidth. 

Several packet scheduling algorithms that support class 
hierarchies with rate guarantees as shown in Fig.~\ref{fig:hierarchy} are available, however, 
deployed or deployable algorithms show significant shortfalls while algorithms without such shortfalls are too complex to be deployable. This paper addresses this dichotomy by presenting a packet scheduler with provable link sharing properties and  low computational complexity. 

For a non-hierarchical setting, a bit-by-bit round-robin algorithm provides link sharing that satisfies  a weighted version of max-min fairness \cite{Book-Bertsekas}. However, bit-by-bit round robin assumes fluid flow traffic and is not implementable as a packet scheduler. 
Weighted-Fair-Queueing (WFQ) \cite{WFQ-keshav} has shown to have a strictly bounded 
deviation from the ideal bit-by-bit round robin \cite{parekh93}. The drawback of WFQ, which extends to some of its approximations \cite{SCFQ2,SFQ}, is that 
it requires to maintain a priority queue that transmits packets in the order of assigned timestamps. 
Deficit-Round-Robin (DRR)~\cite{DRR} is a packet-level round-robin scheduler for  
variable-sized packets, whose link sharing ability is inferior to WFQ, but with a simpler implementation. 
Due to the low complexity,  Linux~\cite{linux-tc} and line-rate 
switches \cite{CRS12000} generally realize link sharing with a round-robin scheduler, such as DRR. 

For class hierarchies as in Fig.~\ref{fig:hierarchy}, Hierarchical Packet Fair Queueing (HPFQ) \cite{HFQ}  
achieves link sharing by employing a cascade of hierarchically organized 
WFQ schedulers.
Packets at the head of the queue of backlogged leaf classes
engage in a virtual tournament, with one round of the tournament 
for each level of the class hierarchy. The tournament starts at the bottom 
of the hierarchy. In each round, the packet with the 
smallest timestamp at one level 
proceeds to the next level. 
The winner of the tournament is 
selected for transmission.
While HPFQ achieves  almost ideal link sharing, it involves a considerable overhead and has not been considered for deployments.\footnote{The claim in \cite{PIFO} of 
realizing HPFQ by a hierarchy of PIFO queues is incorrect, as counterexamples are easily 
constructed when packet sizes are variable.} 

Attempts to extend DRR to a class hierarchy have so far not resulted in practical scheduling algorithms. 
In \cite{HDRR1}, the class hierarchy is mapped to a flat hierarchy by interleaving classes according to their weight guarantees. This results in good fairness properties, but rounds grow prohibitively large which may result in excessive delays between packet transmissions for some classes. 
Other efforts in this direction, e.g.,~\cite{NestedRR,HDRR2} make scheduling decisions in multiple stages, one per level in the class hierarchy, and thus inherit the drawbacks of HPFQ. 

Class-based queuing (CBQ)~\cite{CBQ}  and  Hierarchical Token Bucket (HTB)~\cite{HTB} are two packet schedulers for link sharing  in class hierarchies that are actually deployed, even if the deployment is limited to Linux systems.\footnote{In the appendix, we provide supplemental information on the operation of CBQ and HTB.}

CBQ provides minimum bandwidth guarantees to traffic classes and  distributes excess capacity to backlogged classes. CBQ measures the transmission rate of each class 
to identify traffic classes that are allowed to transmit, which are then served by a  variant of DRR. 
HTB tries to improve the efficiency of CBQ by metering the transmission rates of classes with token bucket filters.  
Classes that exceed their rate guarantee can `borrow' bandwidth from classes further up in the class hierarchy.  HTB schedules packets with a set of DRR schedulers, where only one DRR scheduler is active at a time. 
In addition to link sharing, HTB also enforces rate limits. HTB has become the primary tool for scheduling and shaping of hierarchically structured traffic flows in Linux servers \cite{BwE,Carousel,Eiffel}. 

CBQ and HTB implement rules that dictate when a class with need for additional bandwidth can transmit, however, with the rules it is not possible to determine (a priori) the allocated rates for a given traffic load. In contrast, the outcomes of schedulers such as HPFQ and hierarchical extensions of DRR schedulers  
satisfy a hierarchical version of max-min fairness, which ensures class guarantees as well as isolation between classes in the hierarchy.

Realizing hierarchical link sharing with round-robin schedulers is attractive, since it does not involve packet timestamps and priority queues, but has shown to be challenging. Extensions of DRR to class hierarchies has so far not resulted in a practical scheduling algorithm. 
On the other hand, CBQ and HTB systematically fail to isolate rate guarantees between classes in different branches of the class hierarchy. In particular, they allow classes to manipulate the rate allocations  by 
reassigning rate guarantees in a subtree of the class hierarchy (see Subsec.~\ref{subsec:exp2}). 
Until now, there does not exist a round-robin packet scheduler for class hierarchies that can  
satisfy rate guarantees  while isolating the allocations in different parts of the class hierarchy.

In this paper, we present {\em Hierarchical Link Sharing} (HLS), 
 the first  round-robin scheduler for hierarchical link sharing 
 that ensures rate guarantees and isolation between classes, and that can run at high line rates. The rate allocation achieved by HLS
satisfies a hierarchical version of max-min fairness.  This allocation is strategy-proof, as defined in \cite{DRF}, 
in the sense that classes cannot improve their allocation through 
wrongful representation of their demand or the demand of their sub-classes.  
HLS is a non-hierarchical variant of DRR with a 
time-variable quantum for each class. 

We have implemented HLS as a Linux kernel module \cite{HLS-github}. We present experiments showing that HLS ensures rate guarantees for and isolation between classes, with an overhead that is comparable to other classful scheduler in the Linux kernel. 

\section{Class Hierarchy: Terminology}
\label{sec:hierarchy}

We introduce terminology needed to describe the relationships between classes in a 
class hierarchy. 
Fig.~\ref{fig:hierarchy-def}  depicts a class hierarchy as a rooted tree, where each 
node represents a class. 
The class at the top of the hierarchy, referred to as {\it root}, represents 
a network interface where the scheduling algorithm is active. 
Leaf nodes in the rooted tree 
represent {\it leaf classes}, which are shown as gray circles. 
As stated earlier, all traffic arrivals are mapped to leaf classes. 
Nodes that are neither the root nor a leaf node represent {\it internal classes}.  
If~$\N$ is the set of all classes, we denote by $\L$ and~$\I$, respectively, the leaf classes and internal classes, with $\N = \L \cup \I \cup \{{\rm root}\}$.

For class~$i$ in the figure, the incoming edge connects to its parent class $p(i)$, 
and the outgoing edges connect to its child classes ${\rm child} (i)$. 
Other needed terms, such as ${\rm anc} (i), {\rm sib} (i), {\rm desc} (i)$, and 
${\rm ldesc} (i)$ are defined in Table~\ref{table:hierarchy} and indicated by 
dashed boxes in   Fig.~\ref{fig:hierarchy-def}. 

\begin{figure}[!t]
\centering
\includegraphics[width=4in]{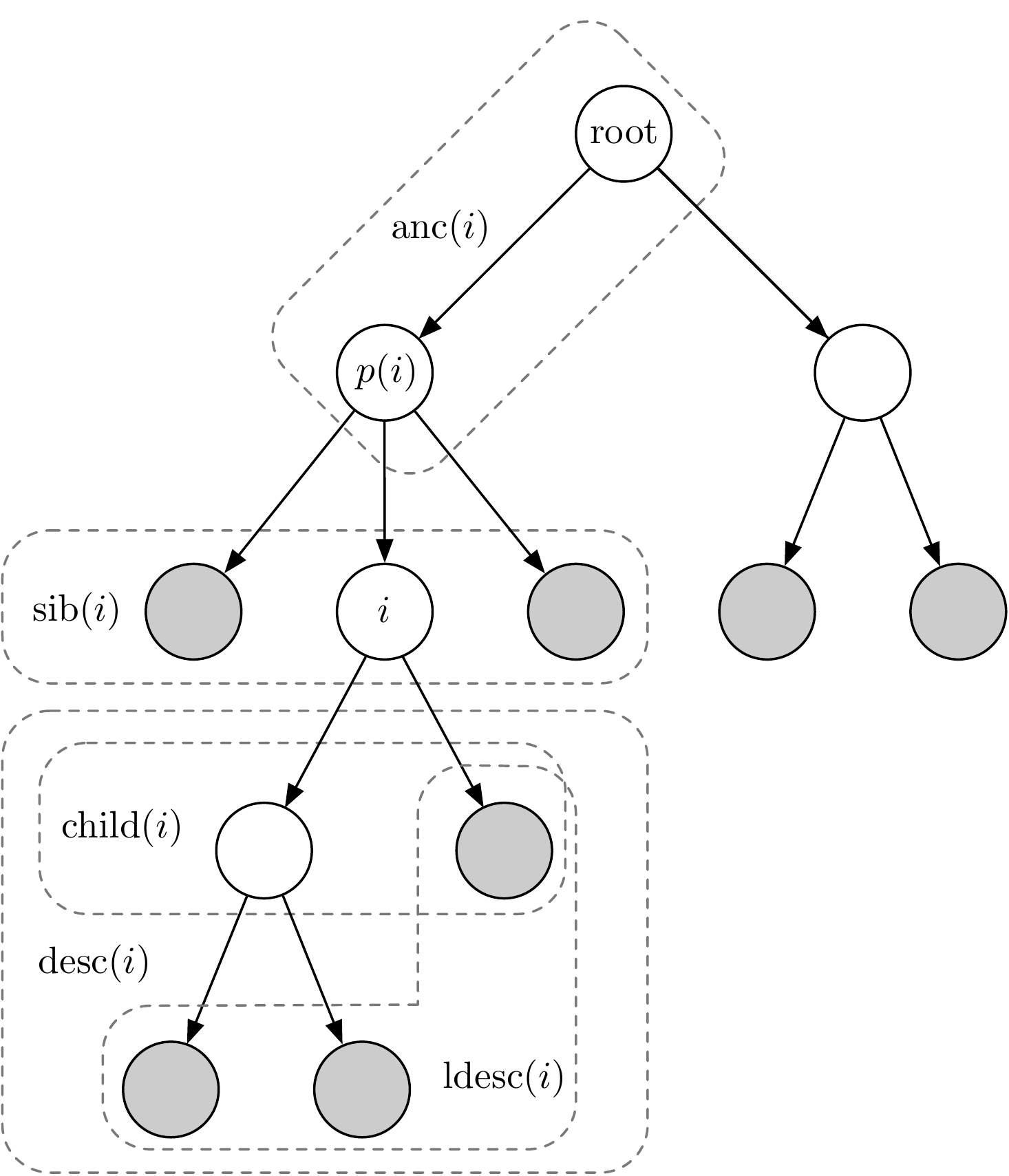}

    \caption{Subsets in link sharing hierarchy.} 
    \label{fig:hierarchy-def}
    
    \vspace{-8pt}
\end{figure}
\begin{table}[!h]
{\small
\begin{center}
\begin{tabular}{| l l l |}
\hline 
$p(i)$ & parent  & Next class on the path from $i$ to root, with  $p ({\rm root})=\emptyset$ \\
${\rm anc} (i)$ & ancestors & Set of classes on  the path from $i$ to root (incl. root) \\
${\rm sib} (i)$ & siblings & All classes with the same parent as  class~$i$ \\
${\rm child} (i)$ & child classes & Set of classes  with  $i$ as parent\\
${\rm desc} (i)$ & descendants & Set of classes with~$i$ as ancestor\\
${\rm ldesc} (i)$ & leaf descendants & Set of leaf classes with~$i$ as ancestor \\
\hline 
\end{tabular}
\end{center}
}
 \caption{Notation for class hierarchy.}
    \label{table:hierarchy}
    
    \vspace{-8pt}
\end{table}

Fig.~\ref{fig:hierarchy-def} is representative of the 
configuration of classful schedulers in Linux.  
In Linux traffic control \cite{linux-tc}, a scheduling discipline is referred to as  a {\it qdisc}. Class hierarchies are specified using configuration commands and 
are built starting from the top of the hierarchy, which is called {\it root qdisc}. Filter expressions are used to map packets to leaf classes. 

\smallskip 
In a class hierarchy, each class is associated with a weight or with a rate guarantee. In Fig.~\ref{fig:hierarchy}, classes
are assigned rate guarantees. Denoting the rate guarantee of class 
$i \in \I \cup \{{\rm root}\}$ by $g_i$, the guarantee must satisfy the 
superadditive property  
\[
g_i \ge \sum_{j \in {\rm child} (i)} g_j \, . 
\]
The guarantee of the root class is the link capacity $C$, that is, $g_{\rm root} = C$. 
There is an alternative specification of link sharing that is based on 
{\it weights}, 
where  $w_i > 0 $ is used to denote  the weight of class $i$. 
If three sibling classes, say classes $1, 2, 3$ have weights $w_1, w_2, w_3$, 
and all siblings are backlogged, the weights indicate that they will split the 
capacity made available to them as a group in the ratio $w_1 : w_2 : w_3$. Viewing link sharing in terms of weights 
is often more convenient, since it allows to express link sharing as dividing   
available bandwidth locally between siblings. In contrast, guarantees appear as global 
quantities with constraints across all classes. We emphasize that the concepts are equivalent.  
Guarantees that satisfy the superadditive constraints above can be viewed as weights, that 
is $w_i = g_i$ for each class $i$. Likewise, given the link capacity~$C$ and weights $w_i$ for 
each class $i$, 
an absolute bandwidth guarantee of class~$i$,~$g_i$,  is computed as 
\[
g_i = \prod_{\substack{j \in {\rm anc}(i) \cup \{i\}\\ j \neq  {\rm root}} }
\frac{w_j}{\sum_{k \in {\rm sib} (j)} w_k} \cdot C \, . 
\]
Fig.~\ref{fig:weights-vs-rates} depicts the relationship between weights and guaranteed rates. 
In the following we will work with weights~$w_i$, but ensure that 
they satisfy $w_i \ge \sum_{j \in {\rm child} (i)} w_j$ for each class~$i$. 
Then, we can use the terms weight and class guarantee interchangeably. 

\begin{figure}[!h]
\centering
\includegraphics[width=3in]{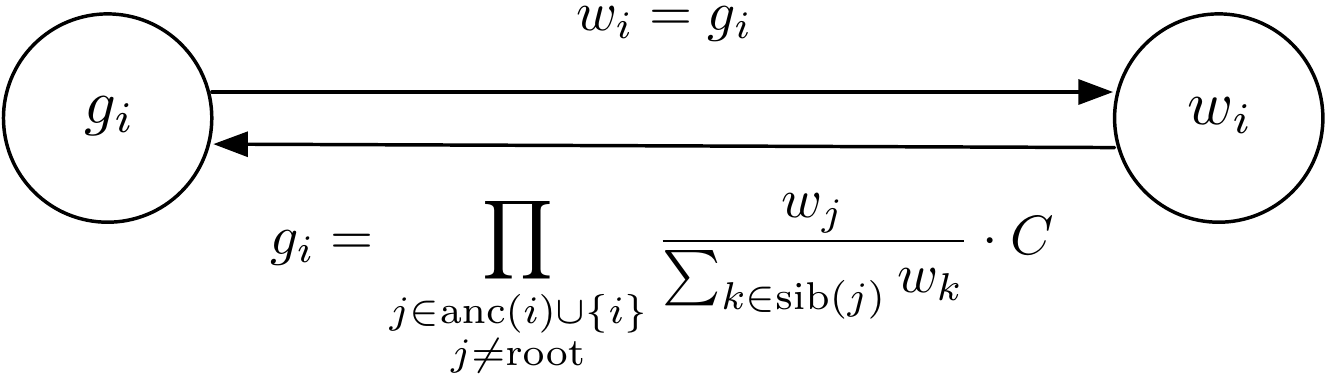}

    \caption{Class guarantees versus weights.} 
    \label{fig:weights-vs-rates}
    
    \vspace{-8pt}
\end{figure}

\vspace{3cm}

\section{Quantifying Link Sharing Goals}
\label{sec:goals}

In a non-hierarchical setting, link sharing between classes can be achieved 
by fair queueing algorithms that approximate a bit-by-bit round robin, 
resulting in a max-min fair rate allocation.  
We  define hierarchical max-min fairness as the result of applying (weighted) bit-by-bit round robin between each group of siblings in a class hierarchy. 
In \cite{HFQ}, such a scheduler is referred to as Hierarchical Generalized Processor Sharing (HPGS). The allocation of this scheduler is also the targeted allocation of HPFQ scheduling. 
Expressions that quantify the solution of this allocation exist for a non-hierarchical setting, 
but are not available for class hierarchies. In the following we quantify both 
the non-hierarchical and hierarchical notions of fairness.

\subsection{Max-min fair allocation}
\label{subsec:mmfair}

We formulate rate allocation for traffic classes 
with fixed-rate traffic at a link with fixed capacity $C$. We define 
\begin{center}
\begin{tabular}{l l}
$r_i$ & Rate request of class $i$, \\
$a_i$  & Rate allocation to class $i$ ($a_i \le r_i$), \\
$w_i$  &  Weight associated with class $i$. 
\end{tabular}
\end{center}
Rate requests are not made explicitly, but are determined by traffic 
arrivals from a class at the link and the resulting backlog. 
In a max-min fair allocation without weights, 
if a class is allocated less than it requests, it receives at least as 
much as any other class. As a consequence, two classes that do not satisfy their 
demand have the same allocation. Also, if the total demand exceeds the 
capacity then the entire link capacity is allocated. 
When specifying weights $w_i$ for each  class $i$, the weighted max-min fair allocation is defined 
by the following rules: 
\begin{enumerate}
\item[{ \it (R1)}] If $a_i < r_i$, then $\displaystyle \frac{a_i}{w_i} \ge \frac{a_j}{w_j}$ for each class $j \in \N$. 
\item[{\it (R2)}]   $\sum_{j\in \N} a_j = \min \bigl( \sum_{j\in \N} r_j , \, C \bigr) $. 
\end{enumerate}

Rule {\it (R1)} states that, if a class is not allocated its entire request, then its allocation 
 in proportion to its weight is as least as large as the (also proportional) allocation 
of any other class. 
The second rule simply ensures that either 
all requests are satisfied or all resources are allocated. 

A weighted max-min fair allocation creates a   
set $\S$ of {\it satisfied classes}, which receive their entire request ($a_i=r_i$). 
and a set $\N \setminus \S$  of {\it unsatisfied classes} with $a_i<r_i$.
Rule {\it (R1)} implies that 
$\frac{a_i}{w_i} = \frac{a_j}{w_j}$ for any two unsatisfied classes. 
The allocation is strategy-proof, since an unsatisfied class cannot increase its allocation by increasing or misrepresenting its request. 

If there is at least one unsatisfied class~$i$, we define the {\em fair share $f$} as 
\[
f = \frac{a_i}{w_i} \, , 
\]
which results in the allocation 
$a_j = \min \{ r_j \, , w_j f \}$. 

Supposing that there exist unsatisfied classes, rule {\it (R2)} yields   
\[
C  = \displaystyle \sum_{j \in \S } r_j  +  \sum_{j \notin \S } w_j f \, . 
\]
Solving for $f$ gives an expression for the fair share as  
\begin{align}
\label{eq:fairshare}
f  = \frac{C - \sum_{j \in \S } r_j}{\sum_{j \notin \S } w_j }  \, .  
\end{align}
As long as there is at least one unsatisfied class, the fair share 
is uniquely defined. Even though the expression for~$f$ is implicit, that is, 
$f$ is defined in terms of $\S$, and  $\S$ is defined in 
terms of $f$, the fair share $f$ can be computed, e.g., by a water filling algorithm as 
given in Algorithm~\ref{alg:wmm_alg}. 
The algorithm uses the fact that, with a fair share $f_n$, each class 
$i$ with $r_i \le    w_i f_n$ is satisfied. 
In the algorithm, the fair share is set to infinity when 
the total demand does not exceed the link capacity. The algorithm computes the fair share iteratively 
by initially assuming that no class with traffic is satisfied, and then labels classes as satisfied until 
the true fair share is obtained. 

\begin{algorithm}[t]
\SetKwRepeat{Do}{do}{while}
\DontPrintSemicolon

	\KwIn{
		Link capacity $C$  and a set $\N$ of classes with $r_i, w_i \ge 0$ for each class~$i \in \N$.}
	\KwOut{Fair share $f$ from~\eqref{eq:fairshare}.}

    \SetKwFunction{WMMMain}{MaxMinFair}
    \SetKwProg{Fn}{Function}{:}{}
    \Fn{\WMMMain{$\N$,~$\{r_i\}_{i\in \N}$,~$\{w_i\}_{i\in \N}, C$}}{
		\uIf{$\sum_{i\in \N} r_i \le C$}{
			\Return{$f \gets \infty$}\;
		}
		\uElse {
		$f_0 \gets 0$\;
		$n \gets 0$\;
		\Do{$f_n \neq f_{n-1}$}{
			$\S^n \gets \left\{i \mid r_i \le  w_i f_n\right\}$\;
			$f_{n+1} \gets \displaystyle \frac{C - \sum_{i\in \S^n}r_i}{\sum_{j\notin \S^n}w_i}$\;
			$n \gets n+1$\;
		}
		\Return{$f \gets f_n$}
		}
		
}
\caption{Computing the fair share~$f$.}
\label{alg:wmm_alg}
\end{algorithm}

\subsection{Hierarchical max-min fair allocation}
\label{subsec:hmmfair}

Next consider a class hierarchy as given in Fig.~\ref{fig:hierarchy-def}.
The requests and allocations of internal classes and the root consist of  
the total requests and allocations, respectively, of their child classes. 
That is, for each $i \in \I \cup {\rm root}$, 
\begin{align}
r_i = \sum_{j \in {\rm child}(i)} r_j \, , 
\quad 
a_i  = \sum_{j \in {\rm child}(i)} a_j    \, . 
\label{eq:aggregate}
\end{align}
With this notation, we can specify a max-min fair allocation 
for class hierarchies. 

A {\it hierarchical weighted max-min fair} ({\it HMM fair}) allocation is defined by these  two rules that hold for each  $i \in \L \cup \I$. 
\begin{enumerate}
\item[{\it (R1)}] If $a_i < r_i$, then $\displaystyle \frac{a_i}{w_i} \ge \frac{a_j}{w_j}$ for all 
$j \in {\rm sib} (i)$. 
\item[{\it (R2)}]   $\displaystyle \sum_{j\in \L} a_j = \min \bigl(\displaystyle \sum_{j\in \L} r_j , \, C \bigr) $. 
\end{enumerate}
The rules are analogous to those for  max-min fairness without a hierarchy. In essence, 
each parent allocates the capacity available to it to its child classes using the max-min fairness principle. 
According to {\it (R1)}, if  a class cannot satisfy its request, then its allocation 
 relative to its weight is at least as large as the allocation  
of any of its siblings relative to the weight of that sibling. 
The second rule makes sure that all available capacity is  utilized.
The allocation is strategy-proof for  each group of siblings, since it 
satisfies 
max-min fairness from Sec.~\ref{subsec:mmfair}, and, therefore, is  strategy-proof for the entire hierarchy. No class can obtain a larger allocation by increasing or misrepresenting its request. 

Using the aggregation in \eqref{eq:aggregate} and rule {\it (R2)}, we can make  two observations for a class~$i$:
\begin{enumerate}
\item[{\it (O1)}] If $a_i < r_i$ then $a_j < r_j$ for all $j \in {\rm anc}(i)$. 
\item[{\it (O2)}] If $a_i = r_i$ then $a_j = r_j$ for all $j \in {\rm desc}(i)$.   
\end{enumerate}
To explain {\it (O1)}, 
if a class receives a smaller rate than it requests, it will try to get more capacity from its parent, which, in turn, will seek to acquire capacity from its own parent, and so forth. Hence, if the request of a class is not satisfied,  the resources of all its ancestors will be exhausted, leaving them unsatisfied as well. Observation {\it (O2)} follows since requests and allocations of an internal class consist of  the 
aggregated requests and allocations of their child classes.  

If the rules for hierarchical max-min fairness are straightforward, the computation of the allocations to classes is much less so. The reason is that the capacity available at an internal class depends on the requests of leaf 
classes in all branches of the hierarchy. The results of an HMM allocation are specified in the following theorem. 

\begin{theorem} 
\label{theo:share-hier}
Given a link with capacity $C$, and a class hierarchy where each class $i$ has a request rate $r_i$ and a weight $w_i > 0$. 
Define $a_{\rm root} = \min \{ r_{\rm root}, C \}$.  \\
Then, the HMM fair allocation 
for each class $i \in \I \cup \L$ is 
\[
a_i = 
\begin{cases}
\min \left( r_i , w_i f_{p(i)} \right)  \, , & {\rm if } \ r_{p(i)} > a_{p(i)} \, , \\
 r_i \, , & \ {\rm otherwise}\, , 
\end{cases}
\]
where the fair share $f_k$ for each $k \in \I \cup \{ {\rm root}\}$ is  
\begin{align}
\label{eq:fairshare-hierarchy}
f_k = \frac{a_k - \sum_{j \in \S_k} r_j}{  \sum_{j \notin \S_k} w_j} \, , 
\end{align}
and where $\S_k$ is defined as 
\[
\S_k = \left\{ j \in {\rm child} (k) \mid r_j <  w_j f_k   \right\} \, . \
\]

\end{theorem} 

\begin{proof}
We proceed by performing an induction over the levels of the hierarchy, starting at the top.
The proof refers to rules {\it (R1)} and {\it (R2)} from  Sec.~\ref{subsec:hmmfair}.

Consider the root class, where we have $a_{\rm root}=\min \{r_{\rm root}, C\}$. If $ r_{\rm root} = a_{\rm root}$, 
then $a_n = r_n$ for every child $n \in {\rm child}({\rm root})$. 
Moreover, with observation {\it (O2)}, we have 
$a_j = r_j$ for all $j \in {\rm desc} ({\rm root}) = \I \cup \L$. 

If $r_{\rm root} >  a_{\rm root}$, 
there exists an $n \in {\rm child}({\rm root})$ with $a_n < r_n$. Define 
$f_{\rm root} = \tfrac{a_n}{w_n}$. By {\it (R1)}, 
for every class $m \in {\rm child}({\rm root})$ with $a_m < r_m$ we have $f_{\rm root} = \tfrac{a_m}{w_m}$. 
We therefore have for each child class $n$ of the root that 
\begin{align*}
a_n = \min \{ w_n f_{\rm root}\, , r_n \}  \, , 
\end{align*}
as well as  
$\S_{\rm root} = \{ i \in {\rm child}({\rm root}) \mid r_i <  w_i f_{\rm root}\ \}$. By rule {\it (R2)}
we get 
\[
\sum_{n \notin \S_{\rm root}} w_n f_{\rm root} + \sum_{n \in \S_{\rm root}} r_n = a_{\rm root} \, , 
\]
which gives~\eqref{eq:fairshare-hierarchy}  for $i={\rm root}$.

For the induction step, we consider a class~$k$ and assume that the 
allocation has been computed for all  ancestors $l \in {\rm anc}(k)$. 
If the allocation of its parent $p(k)$ was such that $a_{p(k)} = r_{p(k)}$, then due to observation {\it (O2)}, we get $a_k = r_k$. 
Now consider $a_{p(k)} < r_{p(k)}$. Then, the parent has computed 
a fair share $f_{p(k)}$ and the allocation to its child class~$k$ was 
$a_k = \min \{ w_k f_{p(k)}\, ,  r_k \}$. If $k$ is a leaf class ($k \in \L$), 
we are done. If $k$ is an internal class, there are two cases. 
If $r_k \le w_k f_{p(k)}$, then $a_k =r_k$, and, by {\it (O2)}, $a_l = r_l$ for 
each $l \in {\rm desc} (k)$. 
If $r_k > w_k f_{p(k)}$, then there exists a class $l \in {\rm child}(k)$
with $a_l < r_l$. 
Defining  
$f_{k} = \tfrac{a_l}{w_l}$,  by {\it (R1)}, we have $a_j =  w_j f_k$ for each   
$j \in {\rm child} (k)$ with $a_j < r_j$. We also define 
$\S_{k} = \{ i \in {\rm child}(k) \mid r_i \le  w_i f_k \}$. 
With rule {\it (R2)} we obtain   
\[
\sum_{n \notin \S_k} w_n f_k + \sum_{n \in \S_k} r_n = a_k \, , 
\]
where $a_k = w_k f_{p(k)}$. Solving the equation for $f_k$ 
gives~\eqref{eq:fairshare-hierarchy} for $i=k$.
\end{proof}

\begin{algorithm}[t]
\SetKwRepeat{Do}{do}{while}
\DontPrintSemicolon

	\KwIn{
		Link capacity $C$,   a set of $\N$ of classes with $r_i, w_i \ge 0$ for each class~$i \in \N$,   and  a class hierarchy $M$ supporting the operations in Fig.~\ref{fig:hierarchy-def}.}
	\KwOut{
		Fair shares~$\{f_i\}_{i\in \I \cup {\rm root}}$ from~\eqref{eq:fairshare-hierarchy}. 
	}
    \SetKwFunction{HMMMain}{HierarchicalMaxMinTopDown}{
    \SetKwProg{Fn}{Function}{:}{}
    \Fn{\HMMMain{$\N$, $M$, $\{r_i\}_{i\in \N}$, $\{w_i\}_{i\in \N}, C$}}{
		\ForEach{$i \in I\cup \{{\rm root}\}$}{
			$r_i \gets \sum_{j\in {\rm ldesc} (i)} r_j$\;
		}
		$L \gets \{{\rm root}\}$\;
		$a_{\rm root}\gets\min(r_{\rm root}, C)$\;

		\Do{$L\neq\emptyset$}{
			$i \gets$\ Select an element from set $L$\;
			$L \gets L\setminus \{i\}$\;

			\uIf{$i \notin  \L$}{
				$f_i \gets$\ \WMMMain{${\rm child}(i)$, $\{r_j\}_{j\in {\rm child}(i)}$, $\{w_j\}_{j\in {\rm child}(i)}$, $a_i$}\;
				\ForEach{$j\in {\rm child}(i)$}{
					$a_j \gets \min(r_j, w_j f_i)$\;
					$L \gets L\cup \{j\}$
				}
			}
		}
		\Return{$\{f_i\}_{i\in \I \cup \{{\rm root}\}}$}
	}
	}
\caption{Computing fair shares for HMM fairness.}
\label{alg:hmm_topdown}

\end{algorithm}

The values of the fair shares in a class hierarchy can be computed with  Algorithm~\ref{alg:hmm_topdown}. 
The algorithm starts at the top of the hierarchy and computes the fair share of the root. Then it proceeds to compute the fair share of children of the root 
and continues to traverse the class hierarchy in a top-down fashion (in no particular order) until a leaf class is reached. The algorithm uses the fair share computation from Algorithm~\ref{alg:wmm_alg}. 

The allocations we discussed are simplified in that arriving traffic and service are expressed in terms of rates. 
Without a hierarchy, expressions for max-min fair allocations for intermittent  bursty traffic exist, 
but they require that traffic be shaped, e.g., by token buckets \cite{parekh93} 
or more general concave bounding functions~\cite{GPS2018}. 

To measure how well a scheduling algorithm realizes an HMM fair allocation for time-variable traffic, we will instead use a fairness metric that measures the deviation from the allocation of an ideal hierarchical bit-by-bit round-robin scheduler.



\section{The HLS scheduler}
\label{sec:hls}

We next present the Hierarchical Link Sharing (HLS) scheduler which allocates rates according to the principle of HMM fairness. We have implemented HLS as a  Qdisc in the Linux kernel \cite{HLS-github}. 

The design of HLS departs from that of HTB and CBQ, 
which both track the transmission rates of classes  
using moving averages in CBQ and token buckets in HTB. If a class 
requires additional bandwidth, both HTB and CBQ allow the class to borrow 
bandwidth from other classes in a greedy fashion.
(HTB and CBQ descriptions prefer the term `borrow,' 
but the bandwidth so acquired is never returned.) 
Different from CBQ and HTB, HLS does not measure the transmission rates of classes.
Instead, it gives transmission quotas to classes such that HMM fairness is satisfied.
Minimum rate guarantees follow as a consequence of achieving HMM fairness.

In HLS, each class~$i$ is associated with an integer 
weight $w_i>0$, which can be set to the rate guarantee of the class 
(see Sec.~\ref{sec:hierarchy}). 

At its core, HLS is a non-hierarchical DRR scheduler with a time-variable 
quantum for each class, which we refer to as {\it quota}. 
Each round of the round robin visits each class that is designated as {\it active}, one by one, in an arbitrary order.  
A leaf class is active if it is backlogged at the start of a round. 
An internal class is active if at least one of its child classes is active. 
We distinguish two kinds of rounds, {\it main rounds} and {\em surplus rounds}, where 
each main round may be followed up by one or more surplus rounds. 
The quota of a class is recomputed and assigned during a visit in a main round. 
If at the end of a main round some classes have unused quota, a surplus round is 
started, where the unused quotas are distributed to classes with a backlog. An additional surplus round  is started if after the completion of a surplus round there is still unused quota left. 

Each active leaf class is visited once per round (main or surplus). 
The determination of the set of active classes 
is done at the start of a round. If a class becomes idle during a round, it remains idle until the end of that  round, even if there 
is an arrival to that class in the middle of the round. 
We use $\leafClasses_{\rm ac}$ and $\internalClasses_{\rm ac}$, respectively, to denote the set of active leaf and internal classes in a main or surplus round.
We also use~$\activeClasses = \leafClasses_{\rm ac}\cup\internalClasses_{\rm ac}$ to denote the set of all non-root active classes.

Each class~$i$ maintains a {\it balance}, denoted by $B_i$, which 
maintains the number of bytes that the class is allowed to transmit (if $i \in \leafClasses$) or that its leaf  
descendants  are allowed to transmit (if $i \in \internalClasses$) in the current round.  
The initial setting is   
\begin{align*}
B_i = 
\begin{cases}
Q^* \, ,  & {\rm if } \  i = {\rm root} \, , \\
0 \, ,  & {\rm otherwise} \, , 
\end{cases}
\end{align*}
where  $Q^* > 0$ denotes the total number of bytes from all classes that can be transmitted in a main or surplus round. 
In Section~\ref{sec:Qstar}, we address how to select $Q^*$.   
In a main round, the root distributes its balance across its child classes, who, in turn, distribute their balance to their own child classes, and so forth. 
The root and active internal classes also maintain a {\it residual}, 
denoted by $R_i$~($i \in \internalClasses_{\rm ac} \cup \{{\rm root}\} $), which contains 
permits for the transmission of bytes that were collected 
from descendants in the previous round, with initial setting $R_i =0$.

In each round, all active internal classes recompute the number of bytes 
that a child class with weight set to one can transmit in the  round, which is referred to 
as the {\it fair quota} and denoted by $F_i$ for class $i$. 
For a class $i \in  \internalClasses_{\rm ac} \cup \{{\rm root}\}$, the fair quota is 
defined as  
\begin{align}
F_i = \left\lfloor\frac{B_{i} }{w^{\rm ac}_i} \right\rfloor \, , 
\label{eq:fair-quota}
\end{align}
where 
\begin{align}
\label{eq:weight-active-children}
w^{\rm ac}_{i} = \sum_{k \in {\rm child} (i) \cap (\leafClasses_{\rm ac} \cup \internalClasses_{\rm ac})} w_k
\end{align}
denotes the sum of the weights of the active child classes of class~$i$. The rounding by the floor function rounds the quantity to an integer multiple of unit byte to avoids floating point operations in the Linux kernel. The root class recomputes $F_{\rm root}$ only in a main round and sets $F_{\rm root}=0$ in all surplus rounds.

Before computing the fair quota, each class~$i \in  \internalClasses_{\rm ac} \cup \{{\rm root}\}$ updates its balance and residual. For the root class the update is 
\begin{align}
\label{eq:balance0}
B_{\rm root}    = B_{\rm root} +  R_{\rm root} \, , \quad  
R_{\rm root}  = 0 \, , 
\end{align}
that is, the residual is added to the balance and then reset. 
For an active internal class, the update is 
\begin{align}
\label{eq:balance1}
\!\!B_i    = B_i +  R_i + w_i F_{p(i)}\, , \ 
B_{p(i)}   = B_{p(i)} -  w_i F_{p(i)} \, ,  \
R_i & = 0 \, . 
\end{align}
Here, the balance of class $i$ is increased by $w_i F_{p(i)}$,  and 
the balance of the parent $p(i)$ is reduced by the same amount. We refer to 
$w_i F_{p(i)}$ as the {\it quota} of class~$i$. 
Also, the residual~$R_i$ is added to the balance and then reset.  
Since the quota of an internal class depends on the fair quota of 
the parent class, the update of balances and computations of the quota is performed in a top down fashion. Without the rounding in~\eqref{eq:fair-quota}, 
every internal class  would have a zero balance 
after the update. With rounding, the remaining balance of a  class~$i \in \internalClasses_{\rm ac} \cup \{{\rm root}\}$ after the update of all its active child classes is bounded by $B_i < w^{\rm ac}_i$. Note that  the unit of $w^{\rm ac}_i$ is in bytes, since it is the remainder of the integer division in~\eqref{eq:fair-quota}. 

Before an active leaf class $i \in \leafClasses_{\rm ac}$ transmits in a round, it performs the update 
 \begin{align}
\label{eq:balance2}
B_i    = B_i  + w_i F_{p(i)}\, , \quad  
B_{p(i)}   = B_{p(i)} -  w_i F_{p(i)} \, . 
\end{align}


\begin{figure}[t]
  \centering
 	\includegraphics[width=0.6\columnwidth]{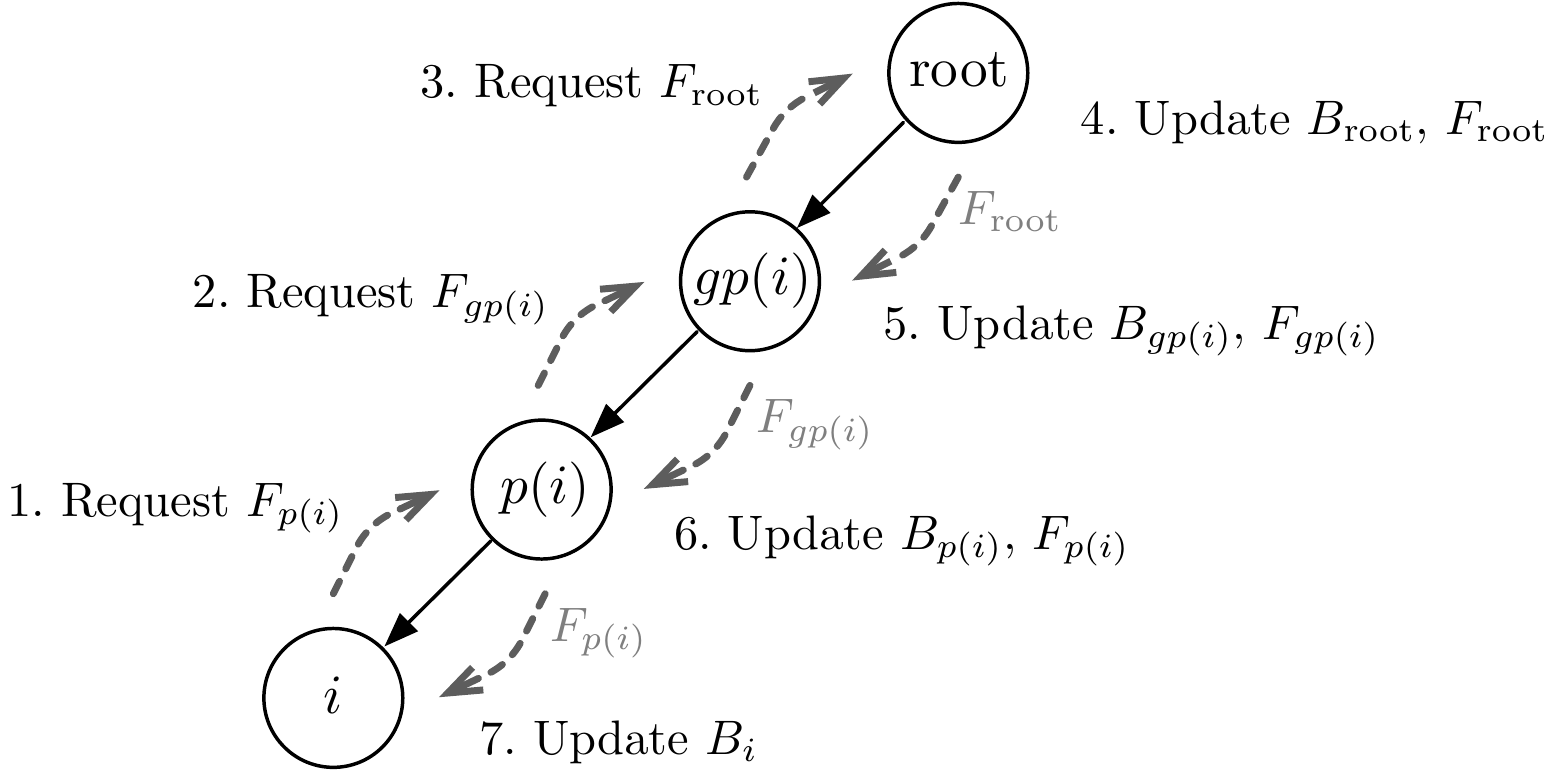}
	
  \caption{Updating fair quotas and balances.}
 \label{fig:update-quota}
\end{figure}
In the HLS Qdisc implementation, the update of  balances in~\eqref{eq:balance0}--\eqref{eq:balance2} 
and the computation of fair quotas in~\eqref{eq:fair-quota}
is initiated by the leaf classes, which is illustrated in  Fig.~\ref{fig:update-quota}. In the figure, node~$i$ represents an active leaf class. When this class is visited in the current round it requests the fair quota~$F_{p(i)}$ from 
its parent. If the parent has not previously computed its quota in the current round, it sends a request for the fair quota~$F_{gp(i)}$ 
to its own parent~$gp(i)$ (we use $gp(i)$ to denote the 
grandparent of class~$i$), and so forth. If the root is reached and the scheduler is in a main round, the balance~$B_{\rm root}$ and the fair quota~$F_{\rm root}$ is computed, and 
then~$F_{\rm root}$ is passed to $gp(i)$. In a surplus round the root returns $F_{\rm root}=0$. 
Next, the internal classes~$gp(i)$ and $p(i)$ use the fair quotas from their 
respective parent to update their balances and compute their own fair quotas. In the last step,  leaf class $i$  updates its balance. 
When the requests~(steps~1--3 in Fig.~\ref{fig:update-quota}) reach a class that already has computed its quota in the current round, no further upstream requests are made. In this fashion, each 
class~$i \in \internalClasses \cup \{{\rm root}\}$  updates its quota only once and updates its balance at most 
$1+ |{\rm child} (i)|$ times per round. 

When an active leaf class $i$ is visited by the round robin, it  updates its balance $B_i$ according to~\eqref{eq:balance2}  before  transmitting packets. 
If the packet at the head of the queue has length $L$ and 
$B_i \ge L$, the packet is transmitted, followed by the update 
\begin{align}
\label{eq:transmit}
B_i   = B_i -  L\, , \quad 
B_{\rm root}  = B_{\rm root} +  L \, . 
\end{align}
By increasing $B_{\rm root}$ for each transmitted packet, the root class accrues a balance  that will be distributed in the next main round.  
Class~$i$ can continue transmitting packets as long as it has 
a sufficient balance. 
If the packet at the head of the queue has size $L$ and $B_i < L$, 
the scheduler turns to the next class in the round robin. 
If a leaf class $i \in \L$ is served and has no more packets to transmit, 
it becomes {\it idle} and returns its balance $B_i$ to its parent 
with the update 
\begin{align}
\label{eq:inactive1}
R_{p(i)}  = R_{p(i)} + B_i \, , \quad  B_i   = 0  \, . 
\end{align}
An internal class becomes idle if all its children are idle.
An idle internal class $i \in \I$ returns its balance $B_i$ and residual $R_i$to its parent by computing 
\begin{align}
\label{eq:inactive2}
R_{p(i)}  = R_{p(i)} + B_i +R_i \, , \quad  B_i   = 0  \, , \quad  R_i   = 0 
 \, . 
\end{align}

Now we see the role of the residual. The residual $R_i$ of an internal 
class $i$ or the root collects  the returned balances from child classes that became idle in the current round. The rationale for not adding the returned balance of an idle child class  immediately  to 
the balance of the parent is to prevent the returned balance from being used in the current round. Doing so would favor leaf classes that are visited later in the round robin. 
By adding the residual to the balance only at the start of a new round,  we ensure that 
all descendants can obtain a portion of the unused balance.

HLS starts a new main round only if the sum of all quotas that has been distributed to  classes has been used for transmissions. Note that a class does not use up its full quota only if it became idle in the current round. This results in  the unused balance (`surplus') being added to the residual of the parent class. If this happens, the residuals accrued in a  round will be distributed to descendant classes in subsequent {\it surplus rounds}. 
The condition to start a surplus round is that at least one internal class $i$  
satisfies $B_i + R_i\ge w^{\rm ac}_{i}$ where the unit of~$w^{\rm ac}_i$ in this case is in bytes instead of unitless, meaning that the class computes a nonzero quota in \eqref{eq:fair-quota} using its balance and residual.
A surplus round operates just like a main round. First, all backlogged classes are marked as active followed by a complete round robin of active classes with the updates from \eqref{eq:fair-quota}--\eqref{eq:inactive2}. The only difference to a main round is that $F_{\rm root}$ is set to zero, meaning that no new quota is distributed from the root. 
If, at the end of a surplus round, there still exists an internal class~$j$ with~$B_j + R_j\ge w^{\rm ac}_j$ another surplus round is started. This continues, until no internal class satisfies the condition, in which case a new main round is started. 

With the updates of the balance counters in~\eqref{eq:balance0}--\eqref{eq:inactive2},  
the sum of  
balances and residuals  of all  classes satisfies 
the invariance  

\begin{equation}
\label{eq:invariant}
\sum_{i \in \allClasses} B_i + \sum_{i \in I \cup\{\rootClass\}} R_i \equiv Q^* \, .
\end{equation}
Since  balances are permits for transmission and the residuals are unused permits for 
transmission, maintaining the invariance ensures that the maximum amount of traffic transmitted in a round does not drift.

\section{Dimensioning of the Round Size}
\label{sec:RoundSize}

The HLS scheduler begins a new main round whenever $Q^*$ bytes have been transmitted. There are two considerations for the selection of $Q^*$. On the one hand, $Q^*$ should not be selected too large, otherwise, the scheduler 
reacts too slowly to changes of the  set of active leaf classes. 
On the other hand, if $Q^*$ is selected too small, the quotas that are 
passed down to  leaf classes may not allow the transmission of any packet. 
In this section we derive a sufficient condition for a lower bound on $Q^*$. 
Our implementation of HLS uses this lower bound to  adjust $Q^*$ dynamically at the start of a main or surplus rounds. 
We will simplify the derivation of the bound by rearranging 
the order in which HLS updates the balance  and transmit packets from classes. Specifically, 
we consider that at the begin of each main or surplus round, HLS updates the quotas and balances  for all classes, before transmitting any packets. As we will show in the next subsection, 
this modification does not alter the transmission schedule of HLS. 

\subsection{Replenishment Phase and Transmission Phase}
\label{sec:hls:phases}

Consider the class hierarchy from \reffig{fig:hierarchy}. Assume that only leaf classes $B2$, $B1$, and $A1$ are backlogged and visited in this order. 
The sequence  of  quota and balance updates and transmissions of these classes in HLS is illustrated in \reffig{fig:hls:replenish_order:1}. In the figure,  quota and balance updates of a class are shown as gray boxes, and  packet transmissions are shown as white boxes. The order follows from the description of HLS in  Sec.~\ref{sec:hls}. Following the recursive process illustrated in Fig.~\ref{fig:update-quota},  a visit of class $B2$ results in  updates first at the {\it root} class, then at $B$, and finally at $B2$. When $B1$ is visited next, 
only the quota and balance of $B1$ are updated, as the quotas of its ancestors ($B$ and {\it root}) have already been computed. When class $A1$ is visited, 
the quotas of $A$ and $A1$ are updated. 
As shown in Fig.~\ref{fig:update-quota}, in this fashion HLS alternates between  quota updates and  packet transmissions as it visits  active leaf classes in a round of the round robin. 

Now consider a modification to HLS, where each main and surplus round starts with a  \emph{replenishment phase}, in which the quotas and balances of all classes are updated according 
to \eqref{eq:fair-quota}--\eqref{eq:balance2}. 
The replenishment phase is then followed by a \emph{transmission phase}, which consists of a round robin that visits each active class and transmits from each 
active class, while updating the balance and residual according to \eqref{eq:transmit}--\eqref{eq:inactive2}. 
We will show that this modification does not alter the order of packet transmissions compared to the  HLS scheduler described in Sec.~\ref{sec:hls}.
\reffig{fig:hls:replenish_order:2} illustrates the updates and transmissions for the example.  
Here, HLS replenishes every active class in the replenishment phase and then visits each active leaf class for the packet transmission in the transmission phase.
We next show that rearranging the quota replenishments and packet transmissions in this manner does not change the behavior of HLS.

\begin{figure*}[ht]

\centering
	\subfigure[Order in HLS.]{
 	\includegraphics[width=0.65\columnwidth]{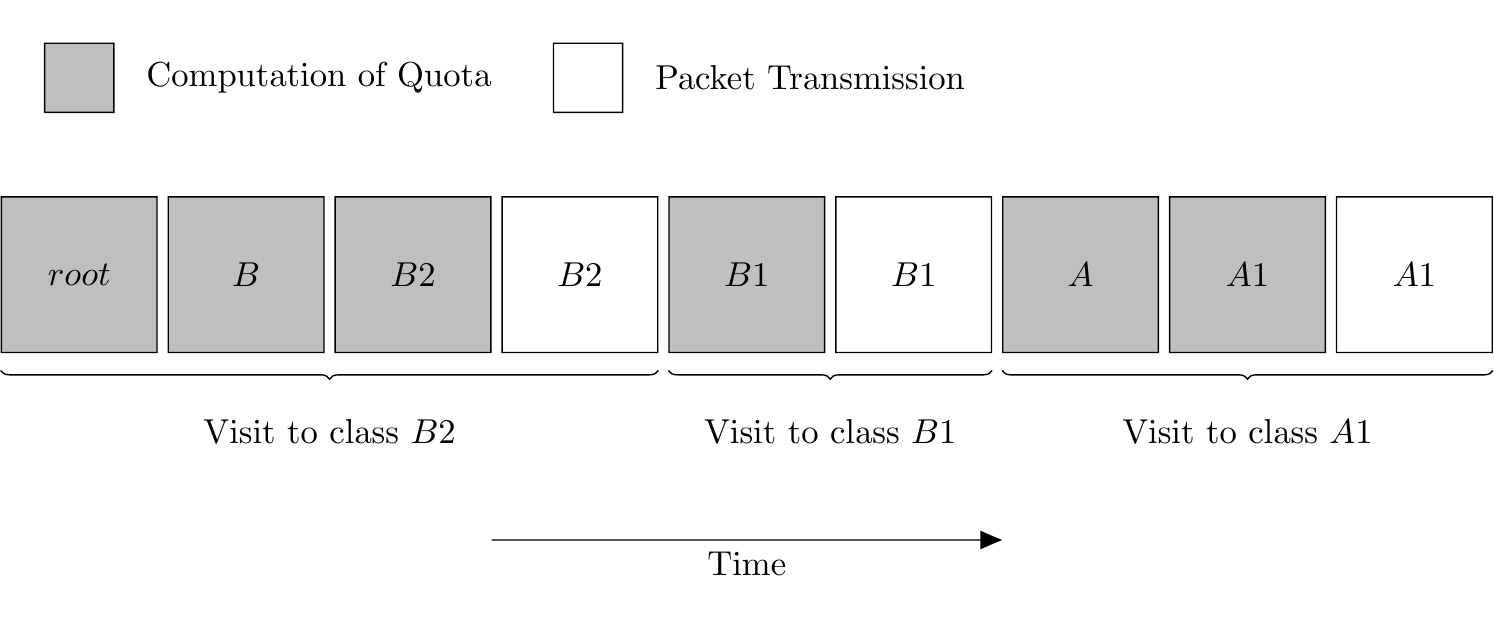}
 \label{fig:hls:replenish_order:1}
	}
	
\centering 	
	\subfigure[Order after rearrangement.]{
 	\includegraphics[width=0.65\columnwidth]{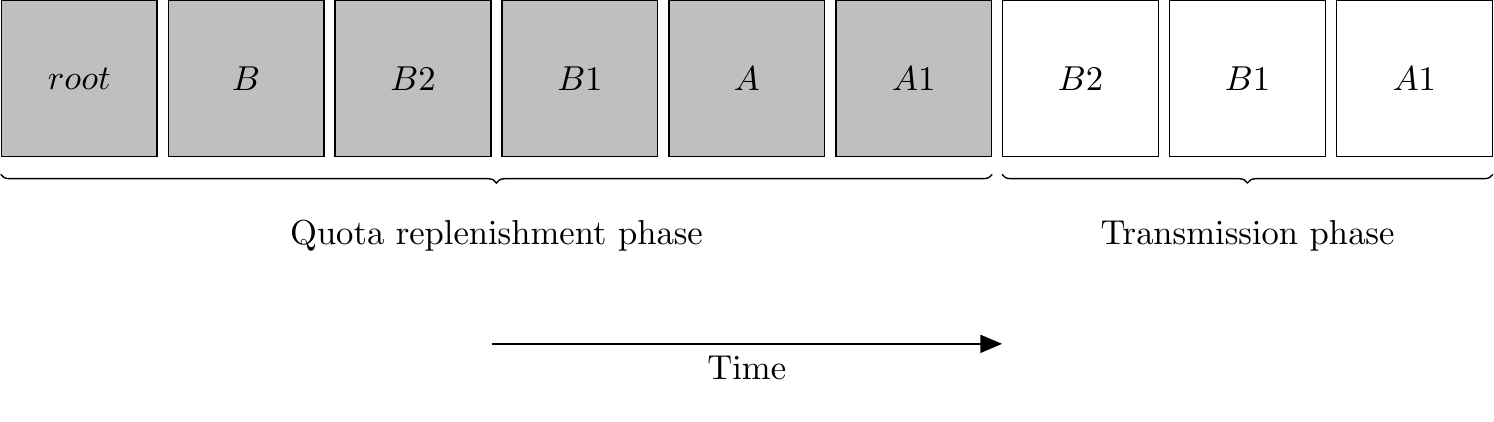}
	  \label{fig:hls:replenish_order:2}
	}   

\caption{Sequence of quota replenishments (gray) and transmissions (white) of classes.}
\end{figure*}

%

\begin{lemma}
\label{lemma:hls:phases}
The modified operation of  HLS with a replenishment phase and transmission phase as described above does not change the order of  packet transmissions in HLS.
\end{lemma}
\begin{proof}
Consider the unmodified HLS scheduler, where the packet transmission of some class~$i$ is immediately followed by the quota replenishment of another class~$j$.
During the packet transmission of class~$i$, only the balances of class~$i$ and its ancestors are updated via \refeq{eq:balance2},~\refeq{eq:inactive1}, and \refeq{eq:inactive2}.
Suppose we swap the packet transmission of class~$i$ and the quota replenishment of class~$j$.
In that case, the only differences observed by the quota replenishment of class~$j$ are the balances of class~$i$ and its ancestors.
Upon computing the fair quota of an ancestor~$k$ of~$j$, if class~$k$ is also an ancestor of class~$i$, class~$k$ is already replenished, and HLS reuses~$\fairQuota k$, which is unchanged by the swapping.
If class~$k$ is not an ancestor of class~$i$, then its balance is not altered and the computed~$\fairQuota k$ remains the same.
So, the quota replenishment of class~$j$ is not affected by the swapping.	

Furthermore, as HLS replenishes class~$i$ before its packet transmission and the quota replenishment of class~$j$, the balance~$B_i$ of class~$i$ remains untouched by the swapping.
Therefore, we can swap the quota replenishment of class~$j$ and the packet transmission of class~$i$ without changing the resulting transmission order of packets.
By repeatedly swapping quota replenishments and packet transmissions of classes, we arrive at the modified HLS scheduler as described above. 
\end{proof}


For the remainder of Sec.~\ref{sec:Qstar} as well as Sec.~\ref{sec:fairness} and Sec.~\ref{sec:gap}, we assume that HLS rearranges the quota replenishments and the packet transmissions into the replenishment phase and the transmission phase. Also, we  use ~$\lmax i$ to denote the maximum packet size of a leaf class~$i\in\leafClasses$.

\subsection{Selection of $Q^*$}
\label{sec:Qstar}
We next present a lower bound on $Q^*$, the number of bytes 
that are transmitted in a main round and subsequent surplus rounds. 
The lower bound ensures that there is at least one packet that can be transmitted in a main round, thus ensuring that the HLS scheduler is work-conserving. As a remark, a surplus round without a transmission is not an issue, since it will be followed by a  main round.

\begin{theorem} \label{theorem:Qstar}
Setting 
\[
Q^* = \sum_{i \in \A} w_i  +  \sum_{i \in \L_{\rm ac}} L^{\rm max}_i  \, , 
\]
ensures that at least one packet can be transmitted in each main round. 
\end{theorem}
Note that {both} terms of the summation depend on the set of active leaf classes.

\begin{proof}
We can ensure that, in any main round, at least one packet can be transmitted 
by satisfying the  condition 
\begin{align}
\sum_{i \in \L_{\rm ac}} B_i \ge \sum_{i \in \L_{\rm ac}} L^{\rm max}_i \, , 
\label{eq:cond}
\end{align}
at the end of the replenishment phase of the round. 
Then, by the pigeon hole principle, there is at least one active leaf class~$i$ with $B_i \ge L^{\rm max}_i$. 

Consider the time at the end of the replenishment phase,  before any packet transmission takes place.  By Lemma~\ref{lemma:hls:phases}, we can perform  the  updates of \eqref{eq:balance0}--\eqref{eq:balance2} for all classes at once.  
Since the residuals of internal classes are set to zero during the updates, we obtain from \eqref{eq:invariant} 
that $\sum_{i \in \N} B_i = Q^*$. Since inactive classes have a balance of zero, positive balances appear only in active classes and the root. With this, we can write~\eqref{eq:cond} as 
\[
Q^*  \ge  \sum_{i \in \I_{{\rm ac}}\cup\{\rootClass\}} B_i  + \sum_{i \in \L_{\rm ac}} L^{\rm max}_i
\, .  
\]
Recall that, after updating the balances of all classes, the remaining balance 
of an  internal class or the root satisfies $B_i <  w^{\rm ac}_i$. Summing up we obtain 
\begin{align*}
\sum_{i \in \I_{{\rm ac}}\cup\{\rootClass\}} B_i  \le
\sum_{i \in \I_{{\rm ac}}\cup\{\rootClass\}} {w^{\rm ac}_i} =
 \sum_{i \in \activeClasses } w_i  \, . 
\end{align*}
Hence, by setting $Q^*$ as given in the theorem, 
we ensure that \eqref{eq:cond} is always satisfied, meaning that  there is at least one packet transmission in each main round.   \end{proof}

We take advantage of the theorem in the Linux Qdisc implementation, where we adjust~$Q^*$ dynamically to the set of active leaf {and internal} classes at the start of each main or surplus round, using the residual~$R_{\rm root}$. When a {leaf} class $i$ becomes active we set $R_{\rm root} = R_{\rm root} + L^{\rm max}_i {+\weight i}$, where $L^{\rm max}_i$ is set to the MTU of 1500 bytes.
When a {leaf} class $i$ becomes idle, we set $R_{\rm root} = R_{\rm root} - L^{\rm max}_i {- \weight i}$.
Similarly, when an internal class becomes active, we set~$R_{\rm root} = R_{\rm root} + \weight i$, and set~$R_{\rm root} = R_{\rm root} - w_i$ when class~$i$ becomes idle.
{These updates to the root class} may result in~$R_{\rm root} <0$ and, after the update of~\eqref{eq:balance0}, 
in $B_{\rm root} < 0$. In this case, we set the fair 
quota of the root to $F_{\rm root}=0$ for the next round. 

\section{Fairness Analysis}
\label{sec:fairness}

To evaluate how well HLS realizes an HMM fair allocation for time-variable traffic, we use a fairness metric that measures the deviation from the allocation of an ideal hierarchical bit-by-bit round-robin scheduler. 
The fairness metric is defined as follows. 

\begin{definition} \label{def:fairness_measure}
    A scheduling algorithm for a class hierarchy is HMM($\alpha$) fair if 
    in an arbitrary time interval~$[t_1, t_2]$ where any two sibling classes~$i$ and~$j$  are backlogged, 
\[
        \left|\frac{D_i(t_1, t_2)}{w_i} - \frac{D_j(t_1, t_2)}{w_j}\right| \le \alpha ,
\]    %
where~$D_i(t_1, t_2)$ is the amount of traffic that class~$i$ or its leaf descendants transmit in the time interval~$[t_1, t_2]$.
\end{definition}
The left-hand side of the equation is the weighted difference between the number of bytes that two classes~$i$ and~$j$ transmit. Since the difference is zero in an ideal hierarchical bit-by-bit round-robin scheduler, the bound $\alpha$ expresses how far a particular scheduling algorithm deviates from 
an ideal link sharing scheduler.

For our analysis of the fairness metric of HLS, we find it useful to use the concept of \emph{subtrees}.
The {\it subtree of class~$i$}, denoted by $\stree{i}$, consists of class $i$ and its descendants, that is, $\stree{i} = \desc{i}\cup\{i\}$. 
We also define the \emph{aggregate balance} of subtree~$\stree i$, denoted by 
~$\aggBalance i$, as  the sum of the balances and residuals of all classes in~$\stree i$, that is, 
\begin{align}
\aggBalance i = \sum_{j\in\stree{i}} \balance{j} + \residual{j} \, ,
	\label{def:aggr-balance} 
\end{align}
For a leaf class~$i$, we obviously get $\stree{i}=i$ and $\aggBalance i= B_i$.

In this section, we assume that $Q^*$ is set to the lower bound given in  Theorem~\ref{theorem:Qstar}.

\subsection{Bound on Fairness of HLS}
\label{subsec:fairness}

We now analyze the fairness metric of HLS. Throughout the analysis we assume that HLS operates as described in Sec.~\ref{sec:hls:phases}, that is, with alternating  replenishment and transmission phases. With Lemma~\ref{lemma:hls:phases}, such a scheduler generates the same transmission schedule as the original HLS. 

The following lemma expresses the number of bytes that a backlogged class~$i$ transmits in a given time interval in terms of the aggregate balance $\aggBalance{i}$. 

\begin{lemma}
	\label{lem:balance_over_time}
	If a non-root class~$i\in\internalClasses\cup\leafClasses$ is backlogged throughout a time interval~$(t_1, t_2)$, then
	\begin{equation}
	\label{nonRootDBound}
		D_i(t_1, t_2) = \aggBalance{i}(t_1) - \aggBalance{i}(t_2) + \weight{i}\sum_{r\in{\rm R}(t_1, t_2)} F^{(r)}_{\parent{i}}
	\end{equation}
	where~$\aggBalance{i}(t)$ is the value of~$\aggBalance{i}$ at time~$t$,~$R(t_1, t_2)$ is the set of main and surplus rounds that are started in the time interval~$(t_1, t_2)$, and~$\fairQuota[(r)]{i}$ is the fair quota of class~$i$ at round~$r$.\\
\end{lemma}
\begin{proof}
	Consider a class~$i\in\internalClasses\cup\leafClasses$ that is continuously backlogged in time interval~$(t_1, t_2)$. The values of~$B_j$ and~$R_j$ for class~$j\in\stree i$ only change in one of the following scenarios: 
	\begin{enumerate}
		\item {Class~$j \in \leafClasses$ transmits a packet of size~$L$:}\\
		 The update for the transmission of a packet follows~\refeq{eq:transmit}, which decreases the value of~$\balance i$ by~$L$, which, in turn, decreases~$\aggBalance i$ by the same amount. So the value of~$A_i$ decreases from this scenario by the total number of bytes transmitted by any class~$j\in\stree i$ during the interval~$(t_1, t_2)$, that is, by~$D_i(t_1, t_2)$.
			
		\item {Quota replenishment of class~$i$ at the start of a main or surplus round:}\\
			For class $i$, the quota replenishment of a round~$r \in R(t_1, t_2) $ increases $\aggBalance{i}$ by~$\weight{i}\fairQuota[(r)]{\parent{i}}$ according to~\refeq{eq:balance1} or \refeq{eq:balance2}. So, the total increase of~$\aggBalance i$ in $(t_1, t_2)$ from this scenario is given by $\sum_{r\in R(t_1, t_2)}\weight i\fairQuota[(r)]{\parent i}$. 
			
		\item {Quota replenishment of class~$j\neq i$ at the start of a main or surplus round:}\\
			The quota replenishment of class~$j$ uses either \refeq{eq:balance1} or \refeq{eq:balance2}.
		In this scenario, the sum of~$\balance j + \residual j + \balance{\parent j}$ remains unchanged, and since~$\parent j\in\stree i$, $\aggBalance i$ also does not change its value.
		\item {Class~$j$ becomes inactive:}\\
		Since we assume that class~$i$ is active throughout the interval,~$j\neq i$. Following~\eqref{eq:inactive1} and~\eqref{eq:inactive2}, $B_j$ is set to zero and the balance of~$\parent j\in\stree i$ is increased by the same amount. This update does not change~$\aggBalance i$.
			
	\end{enumerate}
Considering the changes from every scenario, we have  
	\begin{equation*}
		\aggBalance{i}(t_2) = \aggBalance{i}(t_1) - D_i(t_1, t_2) + \sum_{r\in R(t_1, t_2)}\weight i\fairQuota[(r)]{\parent i}\,.
	\end{equation*}
	We then rearrange the terms to get \refeq{nonRootDBound}.
\end{proof}

Next, we present a set of lemmas that seek to bound the aggregate balance $\aggBalance{i}$.
\begin{lemma}
	\label{lem:max_b_star_time}
	For any time~$t$, let~$t'$ be the time at the end of the replenishment phase of the most recent main round prior to time~$t$. Then, for any child class~$i$ of the root class,
	\begin{equation*}
		\aggBalance i(t')\ge\aggBalance i(t)\,, 
	\end{equation*}
where $\aggBalance i(t')$ and $\aggBalance i(t)$, respectively, are the aggregate balance at times $t'$ and $t$.
\end{lemma}
\begin{proof}
	Without loss of generality we only consider active classes. 	In the time interval~$(t', t]$,  the value of~$\balance j$ or~$\residual{j}$ for~$j\in\stree i$ changes only in one the following situations: 
	\begin{enumerate}
		\item  {Quota replenishment of a child of class~$j$ at the start of a surplus round:}\\
		As~$j\neq\rootClass$, HLS performs a quota replenishment on class~$k\in\child j$, updating~$\balance k$,~$\residual k$, and~$\balance j$ using either \refeq{eq:balance1} or \refeq{eq:balance2}. Since the sum~$\balance k + \residual{k} + \balance j$ does not change by the updates, neither does~$\aggBalance i$.
		
		\item  {Quota replenishment of class~$j$ at the start of a surplus round:}\\
			The quota replenishment of class~$j$ uses either \refeq{eq:balance1} or \refeq{eq:balance2}.
	Since $F_{root}=0$ in surplus rounds, the child classes of the root perform a quota replenishment only at the beginning of a main round, a time which is not in 
	the interval~$(t', t]$. Therefore, we have ~$j\neq i$.   
		For all other classes $j\in\stree i$, the sum of~$\balance j + \residual j + \balance{\parent j}$ remains unchanged, and so does~$\aggBalance i$.
	
			\item {Class~$j \in \leafClasses$ transmits a packet:}\\
		 The update for the transmission of a packet follows~\refeq{eq:transmit}, which decreases the value~$\balance i$, which, in turn, decreases ~$\aggBalance i$.

		\item {Class~$j$ becomes inactive:}\\
		If $j \neq i$, following~\eqref{eq:inactive1} and ~\eqref{eq:inactive2}, $B_j$ is set to zero and the balance of ~$\parent j$ is increased by the same amount. This update does not change~$\aggBalance i$. If $j = i$, then~$A_i$ becomes zero. This follows 
		from~\eqref{eq:inactive1} and~\eqref{eq:inactive2}, and the fact that $p(i) \notin \stree{i}$. 

	\end{enumerate}
	The claim follows since in all cases, the value of~$\aggBalance i$ either remains the same or decreases.
\end{proof}

\reflem{lem:max_b_star_time} implies that the maximum value of $\aggBalance i$ occurs at the end of the replenishment phase of the main round.

\begin{lemma}
	\label{lem:end_of_round_leaf_bound}
	At the end of a main or surplus round, each leaf class~$i\in\leafClasses$ satisfies
	\begin{equation*}
		0 \le \balance{i} {<} \lmax{i}.
	\end{equation*}
\end{lemma}
\begin{proof}
	We only need to consider classes that are active at the start of the main or surplus round. When HLS visits class~$i$ in the transmission phase, HLS transmits the packet at the head of the buffer from class~$i$ as long as~$L_i \le \balance{i}$ where~$L_i$ is the size of the packet at the head of the buffer.
	So, the visit of class~$i$ ends when~$\balance{i} < L_i$ or if class $i$ becomes idle. In the first case, the claim follows from $ L_i \le \lmax{i}$. 
	In the second case, we have $\balance{i}=0$ due to~\eqref{eq:inactive1} and~\eqref{eq:inactive2}. 
\end{proof}

\begin{lemma}
	\label{lem:bound_before_replenishment}
	At the start of a main round, before the replenishment phase, each leaf class~$i\in\leafClasses$ satisfies 
	\begin{equation*}
		0 \le \aggBalance{i}(t) {<} \lmax{i},
	\end{equation*}
	and  each internal class~$i\in\internalClasses$ satisfies 
	\begin{equation*}
		0\le\aggBalance i(t) \le \sum_{j\in {\desc{i}}}\weight{j} + \sum_{j\in\ldesc{i}} \lmax{j}\,,
	\end{equation*}
\end{lemma}
Note that the unit of~$\weight j$ in the equation above is in bytes.

\begin{proof} Consider the time at the start of a main round, which immediately follows the end of the previous main or surplus round. 
	For $i\in\leafClasses$, the lemma 
	follows from \reflem{lem:end_of_round_leaf_bound}.
	For a class  $i \in \I$, at the start of the a round, it must hold that 
	\begin{equation}
		\label{lem:bound_on_balance}
		\balance{i} + \residual{i} < \weight[ac]i\,, 
	\end{equation}
	since otherwise next round will be a surplus round (instead of a main round). Recall that $\balance{i} + \residual{i} \ge \weight[ac]i$ at the end of a main or surplus round is the condition to start a new surplus round. 
	It then follows from~\eqref{eq:weight-active-children} that 
	\begin{align}
		\balance{i} + \residual i
		\le \sum_{k\in\child{i}} \weight{k}\,.\label{eq:beginBound2}
	\end{align}
Then we derive 
	\begin{align*}
		\aggBalance{i} 
		&= \sum_{j\in\stree{i}} \left( \balance{j} + \residual{j} \right)\\
		&= \sum_{j\in\stree{i}\cap\internalClasses} \left(\balance{j} + \residual{j}\right) + \sum_{j\in\stree{i}\cap\leafClasses} \balance{j} \\
		&{<} \sum_{j\in\stree{i}\cap\internalClasses}\sum_{k\in\child{j}}\weight{k} + \sum_{j\in\stree{i}\cap\leafClasses} \lmax{j}\\
		&\le \sum_{j\in{\desc{i}}}\weight{j} + \sum_{j\in\ldesc{i}} \lmax{j}.
	\end{align*}
	The first line uses the definition in~\eqref{def:aggr-balance}.
	We then split the sum over the flows in $\stree{i}$ into internal and leaf classes in the second line, where we use that leaf classes have no residual.
	The first term in the third line follows from \refeq{eq:beginBound2}, and the second term follows by applying \reflem{lem:end_of_round_leaf_bound}.
	We arrive at the last line by rearranging the sums.
\end{proof}

\newcommand{\mainRoundBalanceLimit}[1]{{\tilde A_{#1}}}
We now define~$\mainRoundBalanceLimit i$ as  an upper bound on~$\aggBalance i$ at the beginning of a main round. With~\reflem{lem:bound_before_replenishment}, we have for every  non-root class~$i$,
\begin{equation*}
	\mainRoundBalanceLimit i = \begin{cases}
		\lmax{i} - 1, & i\in\leafClasses,\\
		\sum_{j\in{\desc{i}}}\weight{j} + \sum_{j\in\ldesc{i}} \lmax{j}, & i\in\internalClasses\,.
	\end{cases}
\end{equation*}

We can use this bound to obtain an upper bound on~$\aggBalance i$ for a child class~$i$ of the root class, which holds at all times. 

\begin{lemma}
	\label{lem:bound_balance_rc}
	Every  class~$i\in\child{\rootClass}$  satisfies at all times the bound 
	\begin{equation}
		\label{eq:bound_balance_rc}
		\aggBalance i\le \mainRoundBalanceLimit i + \weight i\max_{j\in\child\rootClass} \left\{\frac{{\weight j + }\mainRoundBalanceLimit j}{\weight j}\right\}\, .
	\end{equation}
\end{lemma}

\begin{proof}
	Consider a child class~$i$ of the root class. Let~$t'$ be the time right before the replenishment phase of a main round and~$t''$ be the time when the quota replenishment is completed. We use   
$\aggBalance i(t')$ and $\aggBalance i(t'')$ to denote the aggregate balance of class~$i$ at these times.  We derive 
	\begin{align}\label{eq:aggBalance_i}
		\aggBalance i(t'')
		&= \aggBalance i(t') + \weight i\fairQuota{\rootClass}\\
		&\le \mainRoundBalanceLimit i + \weight i\left\lfloor\frac{\balance\rootClass(t')}{\weight[ac]{\rootClass}}\right\rfloor\\
		&\le \mainRoundBalanceLimit i + \weight i\frac{\aggBalance \rootClass(t')}{\weight[ac]{\rootClass}}\,.
	\end{align}
	The first line indicates added quota during the replenishment phase.  The second line uses~\reflem{lem:bound_before_replenishment} and the definition of~$\fairQuota{\rootClass}$.
	We then relax the floor function.
	In the last line, we use ~$\balance\rootClass(t') \le \aggBalance \rootClass(t')$ and drop the floor function. 
	
	Next we obtain a bound on $\aggBalance \rootClass(t')$.
	\begin{align*}
		\aggBalance \rootClass(t')
		&= \sum_{j\in\activeClasses}\weight j + \sum_{j\in\leafClasses_{\rm ac}}\lmax{j}\\
		&= \sum_{\substack{c\in\activeClasses\\c\in\child{\rootClass}}}
		 \sum_{\substack{j\in\stree c\\j\in\activeClasses}} \weight j 
		 + \sum_{c\in\child{\rootClass}} \sum_{\substack{j\in\leafClasses_{\rm ac}\\j\in\stree{c}}}\lmax{j}\\
		&\le \sum_{\substack{c\in\activeClasses\\c\in\child{\rootClass}}}\left(\sum_{\substack{j\in\activeClasses\\j\in\stree{c}}}\weight j+\sum_{\substack{j\in\leafClasses_{\rm ac}\\j\in\stree{c}}}\lmax j\right)\\
		&= \sum_{\substack{c\in\activeClasses\\c\in\child{\rootClass}}}\left(\sum_{j\in\stree{c}}\weight j+\sum_{j\in\ldesc c}\lmax j\right) \\
		&\le \sum_{\substack{c\in\activeClasses\\c\in\child{\rootClass}}}{\weight c +} \mainRoundBalanceLimit c\,.
			\end{align*}
	The first line follows because the total quota in the hierarchy is given by~$Q^*$, which we  assumed is set to the lower bound given in Theorem~\ref{theorem:Qstar}.
We split the summation in the second line.
	For the first term of the second line, note that for each~$j\in\internalClasses_{\rm ac}\cup\leafClasses_{\rm ac}$, there is exactly one {(active)} class~$c\in\child\rootClass\cap\activeClasses$ that     is an ancestor of~$j$ ($c \in \anc j$) or is itself class~$j$ ($c=j$).
	Conversely, every active class in~$\stree c$ is also in~$\internalClasses_{{\rm ac}}\cup\leafClasses_{{\rm ac}}$.
	Therefore, the first terms in the first and the second lines are equivalent.
	The second term in the second line follows from similar considerations. 
	In the third line, we  relax the sum  and combine similar terms, and then and rearrange the sums in the fourth line.
	In the last line, we  apply the definition of~$\mainRoundBalanceLimit c$.
	
	With the  bound on $\aggBalance \rootClass(t')$ we continue to derive  
	\begin{align*}
		\frac{\aggBalance \rootClass(t')}{\weight[ac]\rootClass}
		&\le \frac{\sum_{j\in\activeClasses \cap \child{\rootClass}}{\weight j +} \mainRoundBalanceLimit{j}}{\sum_{j\in\activeClasses \cap \child{\rootClass}}\weight j}\\
		&\le \max_{j\in\child\rootClass\cap\activeClasses}\frac{{\weight j +} \mainRoundBalanceLimit j}{\weight j}\\
		&\le \max_{j\in\child\rootClass}\frac{{\weight j +}\mainRoundBalanceLimit j}{\weight j}.
	\end{align*}
	 {Using this inequality, }we obtain for \eqref{eq:aggBalance_i} that 
	\begin{equation}
	\label{eq:bound_balance_rc_i}
		\aggBalance i(t'') \le \mainRoundBalanceLimit i + \weight i\max_{j\in\child\rootClass}\frac{{\weight j + }\mainRoundBalanceLimit j}{\weight j}\,.
	\end{equation}
	Then for any time~$t$, let time~$t'$ be the time at the end of the replenishment phase of the most recent main round prior to time~$t$. From Lemma~\ref{lem:max_b_star_time} and~\eqref{eq:bound_balance_rc_i}, it follows that
	\begin{equation*}
		\aggBalance{i}(t) \le \aggBalance{i}(t') \le \mainRoundBalanceLimit i + \weight i\max_{j\in\child\rootClass}\frac{{\weight j + }\mainRoundBalanceLimit j}{\weight j}\,,
	\end{equation*}
	satisfying~\eqref{eq:bound_balance_rc}.
\end{proof}

We continue with computing an upper bound of~$\aggBalance i$ for classes other than the children of the root class.
We first define 
	\begin{equation*}
		\rc i = \begin{cases}
			i,&i\in\child\rootClass,\\
			\rc {\parent i},&\text{otherwise},
		\end{cases}
	\end{equation*}
which is the ancestor of class~$i$ that is a child of the root class (if $i \not\in {\rm child} ({\rm root})$), or class $i$ itself (if $i \in {\rm child} ({\rm root})$). With this definition we can provide a bound on the aggregate balance of classes that holds at all times. 

\newcommand{\aggBalanceLimit}[1]{{\bar A_{#1}}}
\begin{lemma}
	\label{lem:bound_balance}
	For any non-root class~$i\in\internalClasses\cup\leafClasses$, the bound 
	\begin{equation}
		\label{eq:bound_balance}
		\aggBalance i \le \underbrace{\mainRoundBalanceLimit{\rc i} + \weight{\rc i}\max_{j\in\child\rootClass} \left\{\frac{{\weight j + }\mainRoundBalanceLimit j}{\weight j}\right\}}_{\aggBalanceLimit i:=}
	\end{equation}
	holds at all times. 
\end{lemma}
As indicated in the lemma, we will denote the bound on the right hand side by $\aggBalanceLimit i$.  
\begin{proof}	
Due to ~\eqref{eq:balance0}--\eqref{eq:inactive2}, we have $\balance i \ge 0$ and $\residual i \ge 0$ for each non-root class~$i\in\internalClasses\cup\leafClasses$.  Then, we can conclude 
from~\eqref{def:aggr-balance}  that
$
		A_i \le A_{\rc i}$. 	Since~$\rc i\in\child\rootClass$, the claim follows by applying Lemma~\ref{lem:bound_balance_rc} to~$A_{\rc i}$.
\end{proof}

We now express an upper bound on the fairness metric for HLS.

\begin{theorem}
	\label{thm:hls-fm}
	HLS is HMM($\alpha^{\rm(HLS)}$)-fair with~$\alpha^{\rm(HLS)} = \max_{i\in\allClasses, j\in\sib{i}} \alpha^{\rm(HLS)}_{ij}$, where
	\begin{equation*}
		\alpha^{\rm(HLS)}_{ij} = \frac{\aggBalanceLimit i}{\weight i} + \frac{\aggBalanceLimit j}{\weight j}
	\end{equation*}
\end{theorem}
\begin{proof}
	Starting with Definition~\ref{def:fairness_measure}, we derive for two sibling classes $i$ and $j$ as follows:
	\begin{align*}
		\left|\frac{D_i(t_1, t_2)}{\weight i} - \frac{D_j(t_1, t_2)}{\weight j}\right|
		&= \Bigg|\frac{\aggBalance i(t_1) - \aggBalance i(t_2)}{\weight i} - \frac{\aggBalance j(t_1) - \aggBalance j(t_2)}{\weight i}\Bigg|\\
		&\le \left|\frac{\aggBalance i(t_1) - \aggBalance i(t_2)}{\weight i}\right| + \left|\frac{\aggBalance j(t_1) - \aggBalance j(t_2)}{\weight i}\right|\\
		&\le \frac{\aggBalanceLimit i}{\weight i} + \frac{\aggBalanceLimit j}{\weight j}\,.
	\end{align*}
	In the first term, by  applying \reflem{lem:balance_over_time} to~$D_i(t_1, t_2)$ and~$D_j(t_1, t_2)$, the rightmost term in \eqref{nonRootDBound} cancels out since $i$ and $j$ have the same parent node. The second line arrives from the property of the absolute function that~$|x - y| \le |x| + |y|$.
	We arrive at the last line by applying \reflem{lem:bound_balance} and the fact that~$\aggBalance i \ge 0$.
\end{proof}

\begin{corollary}
	\label{col:hls-fm-flat}
	For a flat hierarchy, i.e., a hierarchy with no internal class, the value of~$\alpha^{\rm(HLS)}_{ij}$ becomes
	\begin{equation*}
		\alpha^{\rm(HLS)}_{ij} = \frac{\lmax{i} - 1}{\weight i} + \frac{\lmax{j} - 1}{\weight j} + 2 + 2\max_{k\in\leafClasses}\frac{\lmax k - 1}{\weight k}\,.
	\end{equation*}
\end{corollary}
This bound is identical to the bound for the non-hierarchical  DRR scheduler with quantum $2 + 2\max_{k\in\leafClasses}\frac{\lmax k - 1}{\weight k}$~bytes~\cite{DRR}, which is no greater than~$2\lmax{}$ where~$\lmax{} = \max_{i\in\leafClasses}\lmax{i}$ is the maximum packet size for the entire scheduler.
For comparison, consider the fairness metric of the HDRR scheduler, which is derived in~\cite{HDRR1}. This is the only available bound available for hierarchical round-robin schedulers. 

\begin{theorem}
	\label{thm:hdrr-fm}
	HDRR is HMM($\alpha^{\rm(HDRR)}$)-fair with $\alpha^{\rm(HDRR)} = \max_{i\in\allClasses, j\in\sib{i}} \alpha^{\rm(HDRR)}_{ij}$, where
	\begin{equation*}
		\alpha^{\rm(HDRR)}_{ij} = \frac{\lmax{}}{\weight i} + \frac{\lmax{}}{\weight j} + Q
	\end{equation*}
	where~$Q$ is the quantum of the HDRR.
\end{theorem}

To compare the bounds of HLS and HDRR, consider the hierarchy in Fig.~\ref{fig:hierarchy}, where we set the weight of a class to its rate guarantee.  
The maximum packet size of each leaf class~$i$ is set to~$\lmax{} = 1500$~bytes. For HDRR, we also set its quantum~$Q$ to the maximum packet size of~$1500$~bytes.
Computing the bounds we obtain 
\begin{align*}
	& \alpha^{\rm(HLS)} = 103.5~\text{bytes}\,, \\
	& \alpha^{\rm(HDRR)} = 1522.5~\text{bytes}\,.
\end{align*}
Here, HLS clearly has a better fairness metric. In general, due to the very different operation of HLS and HDRR, it is not feasible to show that HLS always has a better fairness metric. In fact, for deep hierarchies, the fairness metric of HDRR can be better than that of HLS. 

 \section{Transmission Gap}
 \label{sec:gap}

For hierarchical round-robin scheduling algorithms, we can define a second performance metric  which expresses the elapsed time between  visits of a given class by the  scheduler.   
The transmission gap of a round-robing scheduling algorithm expresses the delay incurred by a 
packet at the head of the transmission buffer at the end of a visit of its class. The transmission gap expresses how long this packet has to wait until the round-robin scheduler 
returns to its class. Note that transmission gap also 
provides a bound on the delay until a class that becomes backlogged is visited by the scheduler. 
As before,  we assume that~$Q^*$ is set to the lower bound given in  Theorem~\ref{theorem:Qstar}.

\newcommand{\transmissionGap}[2][]{{\gamma^{#1}_{#2}}}
\begin{definition}
	A round-robin based scheduling algorithm has a transmission gap of~$\transmissionGap{}$ if after visiting an arbitrary (leaf) class $i$ at time~$t$, which remains backlogged after the visit, the next visit to class~$i$ is guaranteed to occur at or before time ~$t + \transmissionGap{}$.
\end{definition}

The transmission gap  and packet delay are related in the sense that, if~$\transmissionGap{}$ is the transmission gap for a scheduling algorithm, and~$d^{\max}$ is the maximum delay experienced by a packet, then~$d^{\max} \ge \transmissionGap{}$.

For HLS, the transmission gap is as given in the following theorem.  
\begin{theorem}
	\label{thm:bound_tgap_hls}
	An HLS scheduler with  link capacity~$C$ has a transmission gap of
	\begin{equation}
		\label{eq:bound_tgap_hls}
		\transmissionGap[\rm(HLS)]{} = \frac{2}{C} \left( \sum_{i\in\internalClasses\cup\leafClasses}\weight i + \sum_{i\in\leafClasses} \lmax i \right) \,.	
	\end{equation}
\end{theorem}

\begin{proof}
	We assume that the computation time in the quota replenishment phase is negligible and that HLS transmits packets for the entire duration of a main or surplus round.
	
	Consider a main or surplus round~$r$. Let~$t^{{(r)}}_{\rm s}$ be the time at the start of round~$r$ after its replenishment phase, and~$t^{{(r)}}_{\rm e}$ be the time at the end of round~$r$.
	Let~$L_i^{(r)}$ be the number of bytes that a leaf class~$i$ transmits in round~$r$, and let~$L^{(r)}$ be the total number of bytes that HLS transmits in round~$r$.
	That is{,}
	\begin{equation*}
		L^{(r)} = \sum_{i\in\leafClasses^{(r)}_{\rm ac}}L_i^{(r)}\,,
	\end{equation*}
	where~$\leafClasses^{(r)}_{\rm ac}$ is the set of active leaf classes at the beginning of round~$r$.
	Finally, let~$Q^*(t)$ and~$\balance i(t)$ be the values of~$Q^*$ and~$\balance i$ at time~$t$, respectively.
	
	In the interval~$(t^{{(r)}}_{\rm s}, t^{{(r)}}_{\rm e})$, the balance~$\balance i$ for each active leaf class~$i$ is only updated using~\eqref{eq:transmit} due to a packet transmission or using~\eqref{eq:inactive1} when class~$i$ becomes idle. If class~$i$ remains active after the visit in round~$r$, it follows that
	\begin{equation*}
		L_i^{(r)} = \balance i(t^{{(r)}}_{\rm s}) - \balance i(t^{{(r)}}_{\rm e})\,.
	\end{equation*}
	{Otherwise,} class~$i$ becomes idle at some time~$t_{\rm idle} \le t^{{(r)}}_{e}$ {when the visit ends, and}
	\begin{equation*}
		L_i^{(r)} = \balance i(t^{{(r)}}_{\rm s}) - \balance i(t_{\rm idle})\,.
	\end{equation*}
	In both cases,~$L_i^{(r)} \le B_i(t^{{(r)}}_s)$.
	With  ~$Q^*$ set to the lower bound from Theorem~\ref{theorem:Qstar} and with ~$B_i (t) \ge 0$ and~$R_i (t) \ge 0$ at all times $t$, we obtain
	\begin{align*}
	\label{eq:hls-tgap-0}
		L^{(r)}
		&= \sum_{i\in\leafClasses_{\rm ac}} L_i^{(r)} \\
		&\le \sum_{i\in\leafClasses_{\rm ac}} B_i(t^{{(r)}}_{\rm s}) \\
		&\le Q^*(t^{{(r)}}_{\rm s})\\
		&\le \sum_{i\in\internalClasses\cup\leafClasses}\weight i + \sum_{i\in\leafClasses} \lmax i \,.
	\end{align*}
	With the link capacity~$C$  and due to our assumption that HLS transmits packets for the entire duration of round~$r$, we get 
	\begin{equation}
		\label{eq:hls-tgap-i}
		t^{{(r)}}_{\rm e} - t^{{(r)}}_{\rm s} = \frac{L^{(r)}}{C} \le \frac{1}{C} \left( \sum_{i\in\internalClasses\cup\leafClasses}\weight i + \sum_{i\in\leafClasses} \lmax i \right)\,.
	\end{equation}
	
	Consider a class~$i$ that HLS finishes visiting at time~$t^{(r_1)}_v$ in some main or surplus round~$r_1$ and remains backlogged after the visit.
	Let~$r_2$ denote the (main or surplus) round that follows  round~$r_1$.
	Since class~$i$ is backlogged after the visit in round~$r_1$, HLS visits class~$i$ again in  round~$r_2$.
	Let the time that HLS visits class~$i$ in round~$r_2$ be~$t^{(r_2)}_v$.

	From the definition of~$t^{(r)}_s$ and~$t^{(r)}_e$, and from our assumption that the computation time of each replenishment phase is negligible, we obtain the relationship
	\begin{equation*}
		t^{(r_1)}_s \le t^{(r_1)}_e = t^{(r_2)}_s \le t^{(r_2)}_e\,.
	\end{equation*}
	Furthermore, since the time~$t^{(r_1)}_v$ is a time within the transmission phase of round~$r_1$, and~$t^{(r_2)}_v$ is a time within the transmission phase of round~$r_2$, it follows that
	\begin{equation*}
		t^{(r_1)}_s \le t^{(r_1)}_v \le t^{(r_1)}_e \quad\text{and}\quad t^{(r_2)}_s \le t^{(r_2)}_v \le t^{(r_2)}_e\,.
	\end{equation*}
	We then combine the inequalities together to obtain
	\begin{equation}
		\label{eq:hls-tgap-iii}
		t^{(r_1)}_s \le t^{(r_1)}_v \le t^{(r_1)}_e = t^{(r_2)}_s \le t^{(r_2)}_v \le t^{(r_2)}_e\,.
	\end{equation}

	We now compute the upper bound for the duration between the visits to class~$i$ in round~$r_1$ and~$r_2$. That is,
	\begin{align*}
		t^{(r_2)}_v - t^{(r_1)}_v
		&\le t^{(r_2)}_e - t^{(r_1)}_s\\
		&\le \left(t^{(r_2)}_e - t^{(r_2)}_s\right) - \left(t^{(r_1)}_e - t^{(r_1)}_s\right)\\
		&\le \frac{2}C  \left( \sum_{i\in\internalClasses\cup\leafClasses}\weight i + \sum_{i\in\leafClasses} \lmax i \right)\\
		&= \gamma^{\rm(HLS)}\,.
	\end{align*}
	We arrive at the first two lines by applying~\eqref{eq:hls-tgap-iii}.
	The third and last lines follow from~\eqref{eq:hls-tgap-i} and the definition of~$\gamma^{\rm(HLS)}$, respectively.
	As such, after HLS finishes visiting class~$i$ at time~$t^{(r_1)}_v$, HLS guarantees to visit class~$i$ again at time~$t^{(r_2)}_v \le t^{(r_1)}_v + \gamma^{\rm(HLS)}$.

\end{proof}

For comparison, we also compute the transmission gap of  the HDRR scheduler from \cite{HDRR1}.
Similar to DRR, HDRR keeps track of a deficit counter~$d_i$ for each leaf class~$i$.
Initially, the value of~$d_i$ for each leaf class~$i$ is set to zero. 
When HDRR visits a leaf class~$i$ for transmission, it increases~$d_i$ by a fixed quantum~$Q$ and transmits a packet at the head of the queue for that class.
At each transmission for a packet of size~$L$, HDRR subtracts~$L$ from~$d_i$.
The visit to class~$i$ stops when~$d_i < L$ where~$L$ is the size of the packet at the head of the queue.
If class~$i$ becomes idle, HDRR sets~$d_i$ to zero.
Note that these are the only scenarios where the value of~$d_i$ changes.
We now compute the number of bytes that HDRR may transmit for a class~$i$ given the number of times that HDRR visits that class.
\begin{lemma}
	\label{lem:hdrr-visit-size}
	Given an HDRR scheduler with quantum Q. If a leaf class~$i$ is visited~$k$ times in a time interval~$(t_1, t_2)$ and HDRR does not visit class~$i$ at time~$t_1$ or~$t_2$, then
	\begin{equation}
		\label{eq:hdrr-visit-size-0}
		kQ - \lmax i \le D_i(t_1, t_2) \le kQ + \lmax i,
	\end{equation}
\end{lemma}
\begin{proof}
	Let~$d_i(t)$ be the value of~$d_i$ at time~$t$. 
	In HDRR, the deficit counter~$d_i$ for class~$i$ only changes when:
	\begin{enumerate}
		\item {HDRR visits class~$i$}:\\
			Each visit increases~$d_i$ by a fixed quantum~$Q$.
			Since there are~$k$ visits in the time interval~$(t_1, t_2)$, the total increase is~$kQ$.
		\item {Class~$i$ transmits a packet}:\\
			When class~$i$ transmits a packet of size~$L$, it subtracts~$L$ from~$d_i$.
			In $(t_1, t_2)$, $d_i$ is therefore decreased by  the number of bytes that class~$i$ transmits in the interval, which is given by~$D_i(t_1, t_2)$.
	\end{enumerate}
	Combining the two cases, we obtain
	\begin{equation}
		\label{eq:hdrr-visit-size-i}
		d_i(t_2) = d_i(t_1) + kQ - D_i(t_1, t_2)\,.
	\end{equation}
	Since HDRR does not visit class~$i$ at time~$t_1$, there are two possible scenarios:
	\begin{enumerate}
		\item
			{HDRR does not visit class~$i$ at all prior to time~$t_1$}:\\
			In this case, the value of~$d_i(t_1)$ equals its initial value, which is zero.
		\item
			{HDRR visits class~$i$ at least once before time~$t_1$}:\\
			Let~$t'$ be the time right after the most recent visit to class~$i$ prior to time~$t_1$, and so~$d_i(t_1) = d_i(t')$ because there is no change to~$d_i$ in the interval~$(t', t_1)$.
			If class~$i$ becomes idle after the visit ends at time~$t'$, then~$d_i(t') = 0$.
			Otherwise, class~$i$ remains backlogged after time~$t'$, and it holds that~$0 \le d_i(t') < L$ where~$L$ is the size of the head of the queue at time~$t'$.
	\end{enumerate}
	In both cases, it follows that
	\begin{equation}
		\label{eq:hdrr-visit-size-ii}
		0 \le d_i(t_1) \le \lmax i\,.
	\end{equation}
	By using the same consideration for~$t_2$, we obtain
	\begin{equation}
		\label{eq:hdrr-visit-size-iii}
		0 \le d_i(t_2) \le \lmax i\,.
	\end{equation}
	We then combine~\eqref{eq:hdrr-visit-size-i},~\eqref{eq:hdrr-visit-size-ii}, and~\eqref{eq:hdrr-visit-size-iii} to arrive at~\eqref{eq:hdrr-visit-size-0}.
\end{proof}

In order to compute the transmission gap of HDRR, we also make use a result from~\cite{HDRR1}.
\begin{lemma}[\hspace{-0.01em}\cite{HDRR1}]
	\label{lem:hdrr-visit-frequency}
	For an HDRR scheduler, if class~$i$ is backlogged, then class~$i$ is visited at least once every~$\beta_i$ visits where
	\begin{equation*}
		\beta_i = \prod_{j\in\anc{i}}\sum_{k\in\child j}\weight k\,.
	\end{equation*}
\end{lemma}

\begin{theorem}
	\label{thm:bound_tgap_hdrr}
	An HDRR with link capacity~$C$ and quantum~$Q$ has a transmission gap of
	\begin{equation*}
		\transmissionGap[\rm \rm(HDRR)]{} = \frac{1}{C} \left(\max_{i\in\leafClasses}(\beta_i - 1)Q + \sum_{j\in\leafClasses}\lmax j \right)\,,
	\end{equation*}
where $\beta_i$ is as given in Lemma~\ref{lem:hdrr-visit-frequency}.
\end{theorem}
\begin{proof}
	For any backlogged leaf class~$i$, let~$t_1$ be the time at the end of a visit to class~$i$ where class~$i$ remains backlogged after the visit, and let~$t_2$ be the time of the subsequent visit to class~$i$ after time~$t_1$.
	Let the~$V$ be the set of visits between~$t_1$ and~$t_2$, and~$M$ be the set of classes that are visited in the time interval~$(t_1, t_2)$. Furthermore, let~$V_j\subset V$ be the set of visits to class~$j\in M$ in the interval~$(t_1, t_2)$. From the definition,~$\{V_j\}_{j\in M}$ is pairwise disjoint, and
	\begin{equation*}
		V = \bigcup_{j\in M} V_j\,.
	\end{equation*}

	Let~$D(t_1, t_2)$ be the number of bytes that HDRR transmits in the interval~$(t_1, t_2)$.
	Since HDRR has the link capacity of~$C$ and it is guaranteed to be backlogged in the time interval~$(t_1, t_2)$ due to class~$i$ being backlogged in the interval, it follows that
	\begin{equation}
		\label{eq:bound_tgap_hdrr-i}
		D(t_1, t_2) = C(t_2 - t_1)\,.
	\end{equation}
	Consider,
	\begin{align*}
		D(t_1, t_2)
		&= \sum_{j\in M}D_j(t_1, t_2) \\
		&\le \sum_{j\in M}|V_j|Q + \lmax j\\
		&= |V|Q + \sum_{j\in M} \lmax j\\
		&\le (\beta_i-1)Q + \sum_{j\in\leafClasses} \lmax j\,.
	\end{align*}
	The first line comes from the definition of~$D(t_1, t_2)$ and the fact that HDRR only visits classes in~$M$ in the interval~$(t_1, t_2)$.
	We then apply Lemma~\ref{lem:hdrr-visit-size} in the second line and rearrange the summation in the third line.
	We then apply Lemma~\ref{lem:hdrr-visit-frequency} and the fact that~$M\subseteq \leafClasses$ in the last line.
	By applying~\eqref{eq:bound_tgap_hdrr-i} and the definition of~$\gamma^{\rm HDRR}$, we obtain
	\begin{align*}
		t_2
		&\le t_1 + \frac{(\beta_i - 1)Q + \sum_{j\in\leafClasses}\lmax j}C\\
		&\le t_1 + \gamma^{\rm(HDRR)}\,.
	\end{align*}
	Therefore, after class~$i$ is visited at time~$t_1$, it is guaranteed to be visited again at time~$t_2 \le t_1 + \gamma^{\rm(HDRR)}$.
\end{proof}

For the the class hierarchy in Fig.~\ref{fig:hierarchy}, with a maximum packet size~$\lmax{} = 1500$~bytes for all classes, and a link with rate~$C = 1$~Gbps, we obtain %
\begin{align*}
	& \transmissionGap[\rm(HLS)]{}
	= 0.146~\text{ms}\,, \\
	&	\transmissionGap[\rm(HDRR)]{}  = 3600~\text{ms}\,.
\end{align*}
We observe that HLS has a significantly smaller transmission gap compared to HDRR.
The difference of the transmission gaps is exacerbated with a larger class hierarchy. 
Consider a class hierarchy consisting of a complete binary tree with~$l$ levels where each left child class has its weight set to~$3$ and the right child class has its weight set to~7.
We again assume ~$C = 1~\text{Gbps}$ and~$\lmax{}=1500$~bytes. 
Table~\ref{table:gap-compare} shows the transmission gaps of HLS and HDRR for the 
range $4 \le l \le 11$. It is apparent that the transmission gap of HDRR becomes unsustainable for large class hierarchies. 
\begin{table}[h]
\centering
\begin{tabular}{l|c|c|c|c|c|c|c|c}
  \\[-1em]
  $l$ & 4 & 5 & 6 & 7 & 8 & 9 & 10 & 11\\[2pt]
 \hline
   \\[-1em]
 $|\leafClasses|$ & 8 & 16 & 32 & 64 & 128 & 256 & 512 & 1024 \\[2pt]
 \hline
  \\[-1em]
$\gamma^{\rm (HLS)}$ (in ms) & 0.096 & 0.19 & 0.39 & 0.77 & 1.55 & 3.09 & 6.18 & 12.37 \\[2pt]
 \hline \\[-1em]
 $\gamma^{\rm (HDRR)}$ (in ms)& 12 & 120 & 1200 & 12,000 & {120,002} & {1,200,003} & {12,000,006} & {120,000,012}
\end{tabular}

\vspace{4pt}

\caption{Transmission gaps of HLS and HDRR for a binary tree hierarchy with $l$ levels.}
\label{table:gap-compare}
\end{table}

\section{Evaluation}
\label{sec:eval}

We have implemented HLS as a kernel module in  Linux  kernel 4.15.0-101-generic {\cite{HLS-github}. A description of 
the implementation is available in \cite{Natchanon-thesis}.}   Here 
we present measurement experiments of 
the HLS Qdisc in Linux and compare them with measurements of the existing link sharing schedulers in Linux, CBQ and HTB. 
The experiments are conducted on Emulab~\cite{emulab}, a network testbed for network experiments.

\subsection{Experimental Setup}
\label{subsec:exp-setup}
The topology of the experiments involves three Linux servers as shown in Fig.~\ref{fig:exp_topo},
designated as \emph{traffic generator}, \emph{scheduler}, and \emph{traffic sink}.
Each server is a Dell PowerEdge R430 with two 2.4 GHz 8-Core CPUs,
64 GB RAM, a dual-port/quad-port 1GbE PCI-Express NICs, and a dual-port/quad-port Intel X710 10GbE PCI-Express NICs. 
The servers run Ubuntu 18.04LTS. 

\begin{figure}[!h]
    \centering
    \includegraphics[width=0.7\columnwidth]{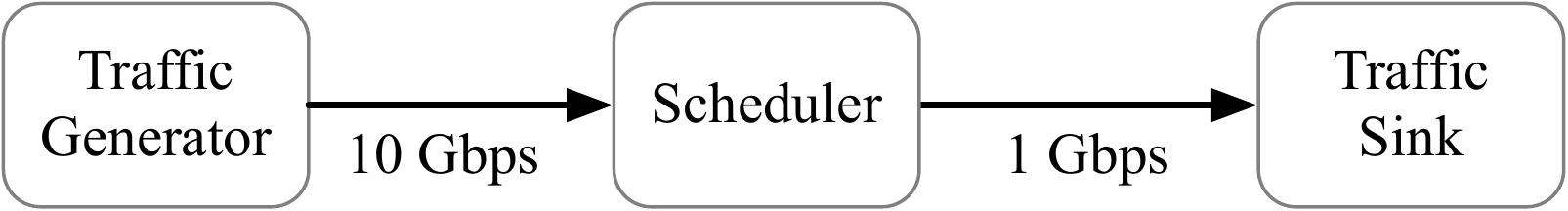}
    \caption{Network topology for experiments.}
    \label{fig:exp_topo}
\end{figure}

The traffic generator and  the scheduler are connected by a 10~Gbps Ethernet link, and the scheduler and  the traffic sink are connected by  a 1~Gbps Ethernet link. Routing tables of all servers are set up statically so that all traffic is routed from the traffic generator to the 
traffic sink. 
The link sharing schedulers are configured at the egress of 
the 1~Gbps interface at the  scheduler node. The traffic generator uses FIFO 
scheduling.  
With this setup we can saturate the outgoing link at the scheduler without overloading its CPUs. 

The traffic generator sends UDP/IPv4 datagrams with a length of 1000 bytes, 
where destination port numbers are mapped to classes at the scheduler node.
Our graphs plot the transmission rates of traffic classes using jumping windows  with  length  $0.2$~s. 
The rate at which the traffic generator sends packets is such that 
it ensures that each active leaf class is permanently backlogged at the egress 
of the scheduler node.


\subsection{Experiment 1: Validation of HMM fairness}
\label{subsec:exp1}


In this experiment, which is a scaled version of an experiment in \cite{CBQ,floyd98}, we show that HLS quickly converges to an HMM fair allocation when the set of active flows changes. 
The class hierarchy of the experiment is as shown in Fig.~\ref{fig:hierarchy} for a 1~Gbps link, but with the following rate guarantees:

\begin{center}
\begin{tabular}{l | c c c c c c}
Class: 			&  	$A$ & $B$ & $A1$ & $A2$ & $B1$  & $B2$\\
\hline 
Rate guarantee:	& 	700	& 300 & 300	 & 400  & 100  & 200 
\end{tabular}
\end{center}


In the experiment, all leaf classes are initially active and transmit packets with a fixed packet size of 1000~B. At certain time intervals, one class becomes idle, in the following sequence:

\begin{center}
\begin{tabular}{l | c c c}
Interval (in seconds): &  $[4,7]$ & $[11,14]$ & $[19,22]$ \\
\hline 
Inactive class: 			& 		$B1$	& $A1$		& $A2$ 
\end{tabular}
\end{center}

\begin{figure}

	\subfigure[HMM fair allocation.]{
 	\includegraphics[width=0.46\columnwidth]{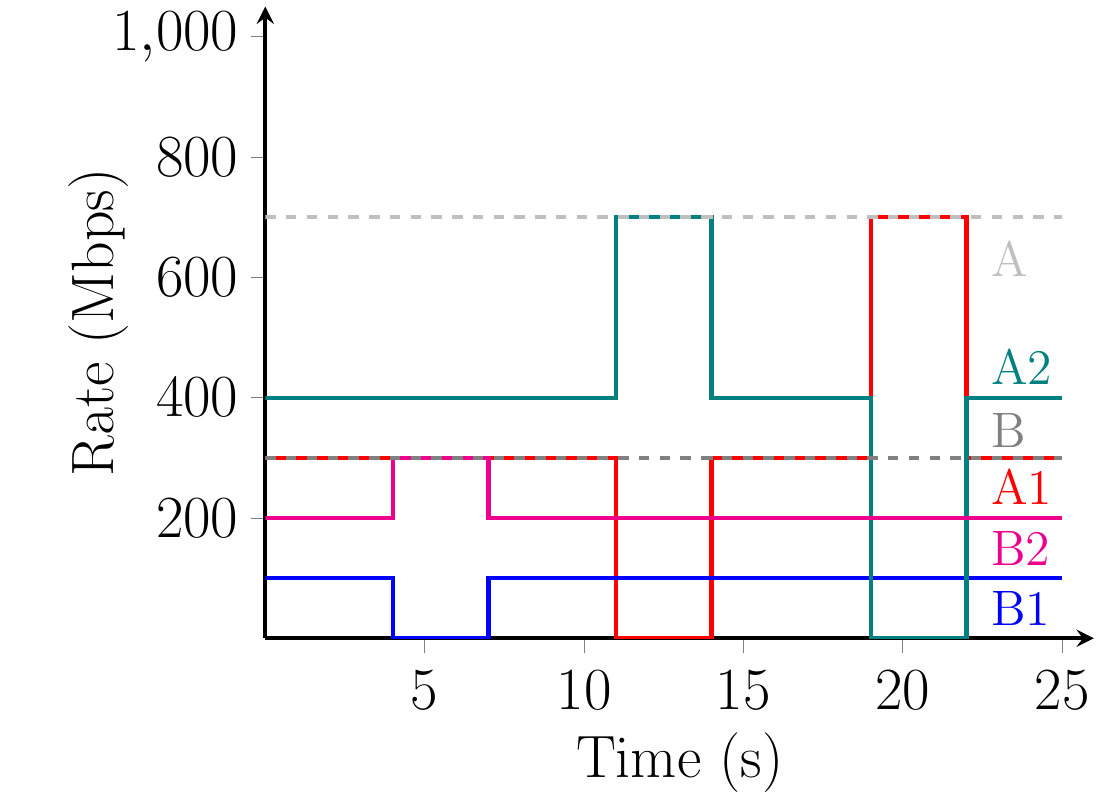}
	\label{fig:cbq-validate-fairness}
	}   
\subfigure[HLS Qdisc (Linux).]{
 	\includegraphics[width=0.46\columnwidth]{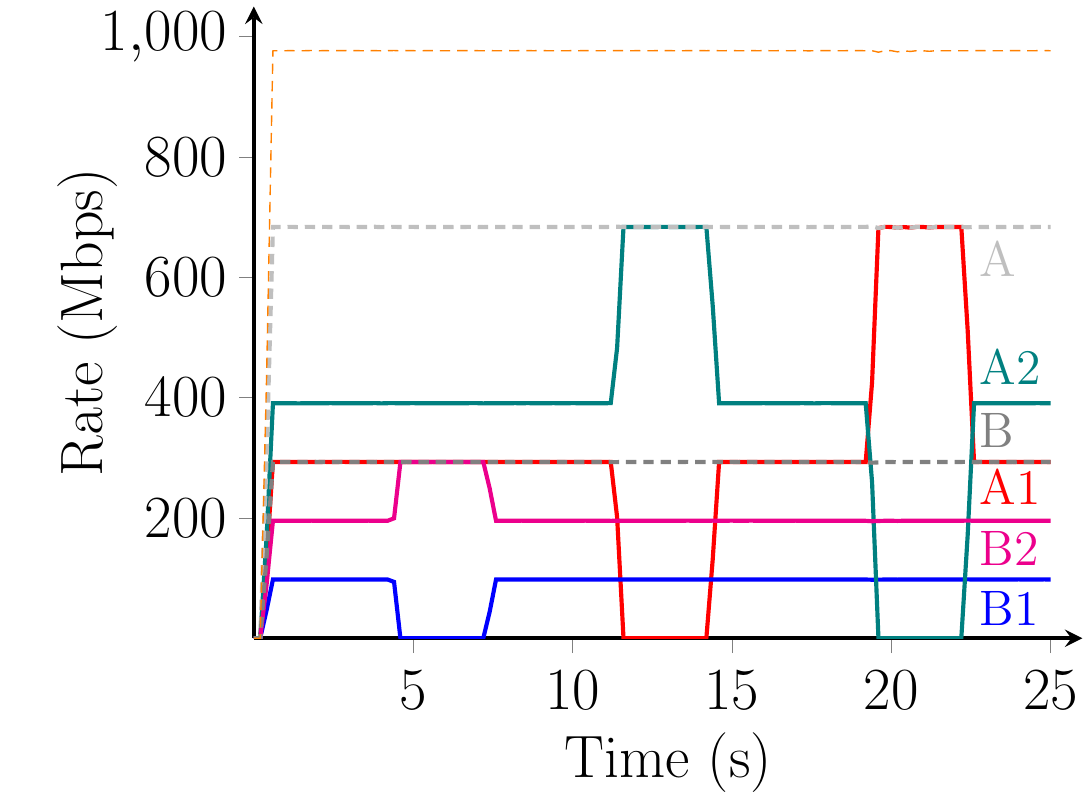}
\label{fig:hls-validate}
}

  \caption{Experiment 1: HMM fairness. }
 \label{fig:Exp1-validate}
%
%

\centering
	\subfigure[CBQ qdisc (Linux).]{
 	\includegraphics[width=0.46\columnwidth]{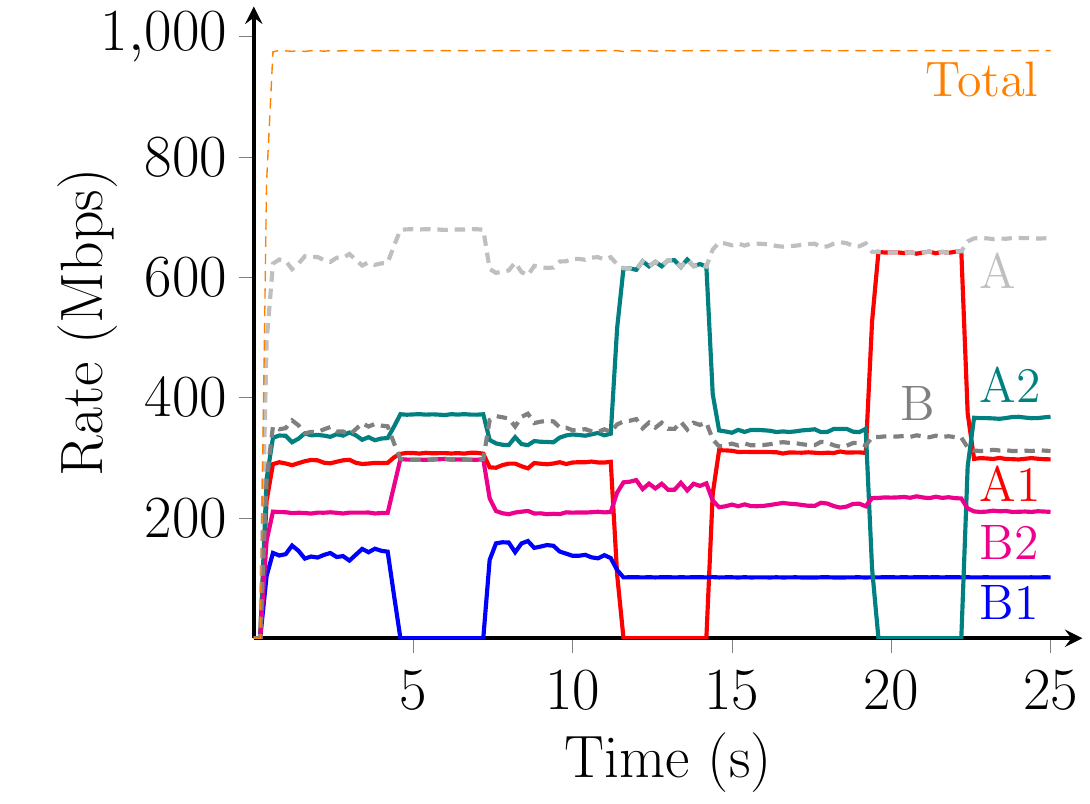}
	\label{fig:cbq-validate-emulab}
	}
%
	\subfigure[CBQ {\it ns2} simulation (formal link sharing).]{
 	\includegraphics[width=0.46\columnwidth]{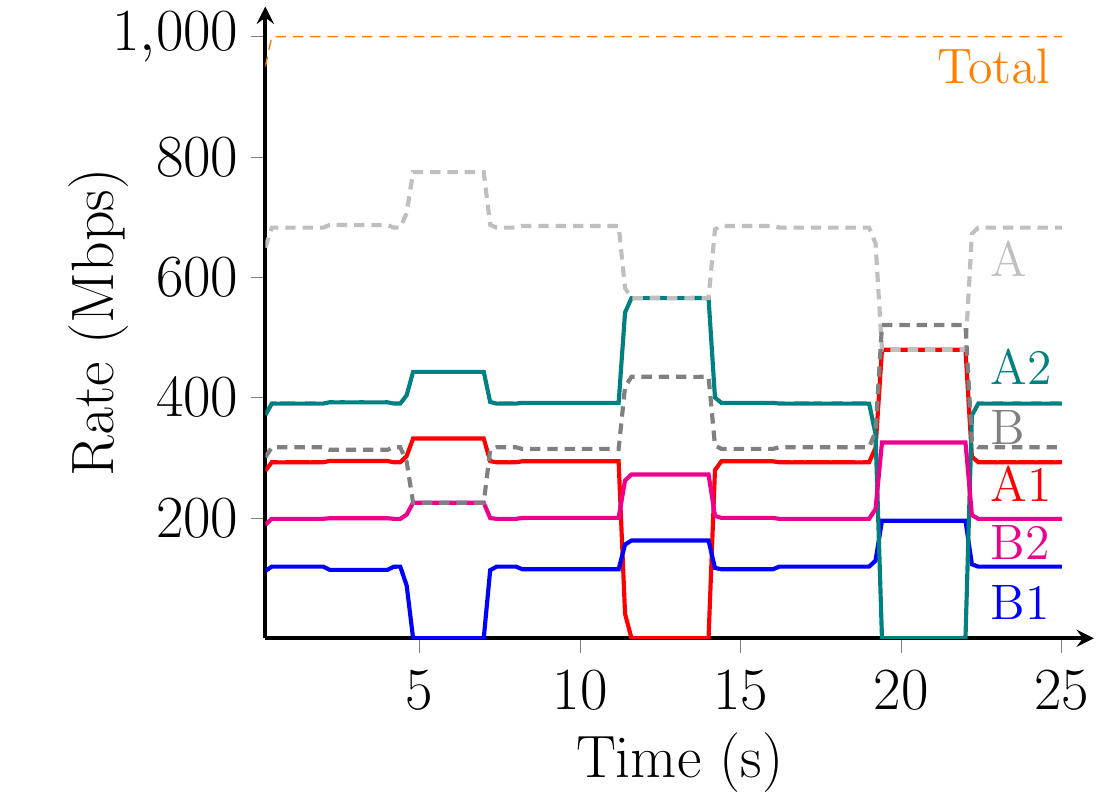}
	\label{fig:cbq-validate-ns2-formal}
	}   

\centering
%

	\subfigure[CBQ {\it ns2} simulation (top-level).]{
 	\includegraphics[width=0.46\columnwidth]{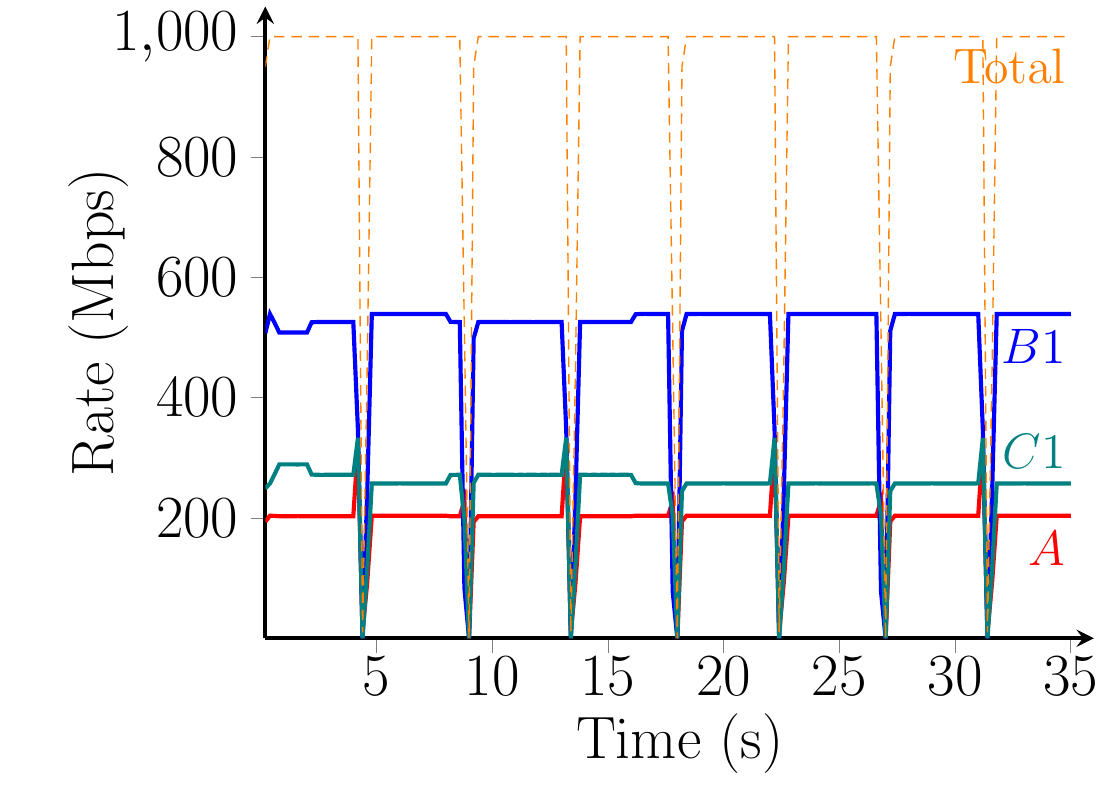}
	\label{fig:cbq-validate-ns2-toplevel}
	} 
	\subfigure[HTB qdisc (Linux).]{
 	\includegraphics[width=0.46\columnwidth]{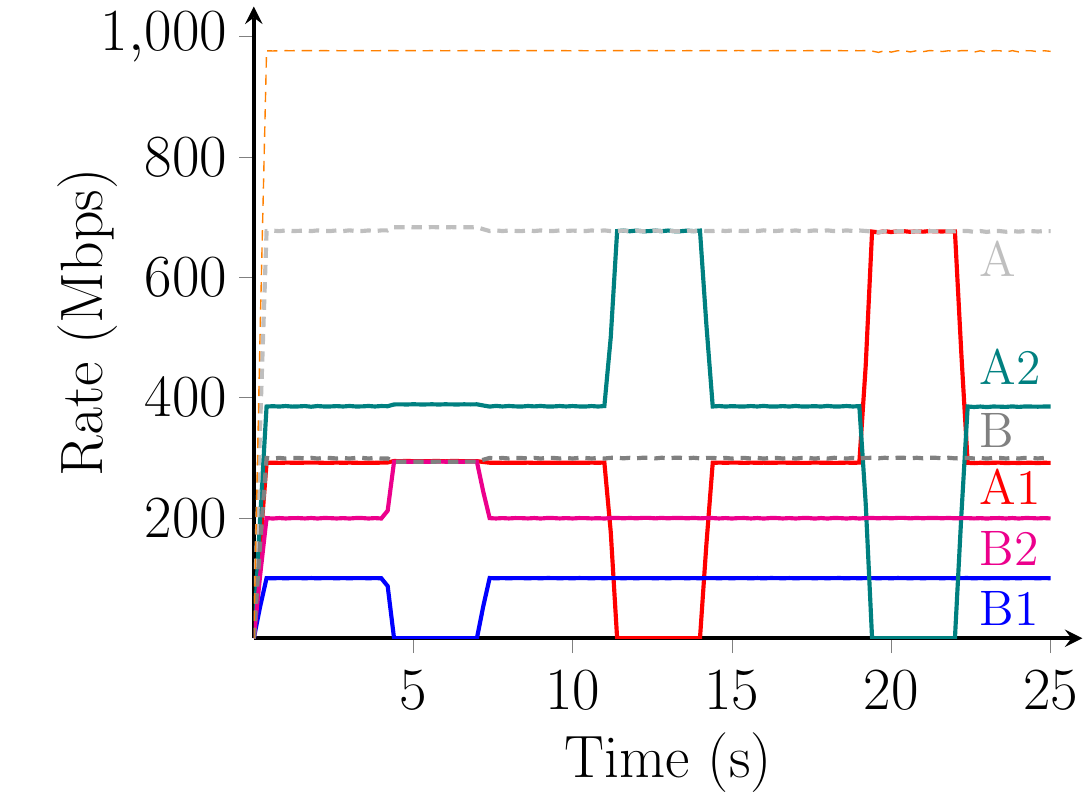}
 \label{fig:htb-validate}
}  
	
  \caption{Experiment 1: HMM fairness (CBQ and HTB).}
\end{figure}

\medskip
The throughput of the classes is shown in Fig.~\ref{fig:Exp1-validate}. The plot for each class is labeled.  
Dashed gray lines show the aggregate traffic of  the internal classes~$A$ and~$B$. The dashed line with label `Total' indicates the aggregate traffic from all classes, which is at or close to the link capacity of 1~Gbps.

Fig.~\ref{fig:cbq-validate-fairness} shows the HMM fair allocations from Theorem~\ref{theo:share-hier}. 
The measured rates for HLS in Fig.~\ref{fig:hls-validate} show that HLS  satisfies  
HMM fairness for all classes at all times. 
When all leaf classes are active, they each obtain their class guarantees. 
If one class drops out, the sibling class consumes the guarantee of 
its sibling. 

Figs.~\ref{fig:cbq-validate-emulab}--\ref{fig:cbq-validate-ns2-toplevel} show the 
results of this experiments for CBQ. Fig.~\ref{fig:cbq-validate-emulab} has the measurements of 
the CBQ Qdisc from Emulab. Figs.~\ref{fig:cbq-validate-ns2-formal} and~\ref{fig:cbq-validate-ns2-toplevel} show the results of {\it ns2} simulations. 
using the  variants  formal link sharing and top-level. 
The data shows that each variant satisfies the guarantees of leaf classes 
at all times. The guarantees of the internal classes $A$ and $B$ are satisfied when all leaf classes are sending traffic. When this is not the case, 
the throughput of one of the internal classes  may fall below its guarantee,   while the other class grabs the remaining bandwidth. Interestingly, 
the Linux qdisc implementation shows smaller violations of  
rate guarantees than the {\it ns2} simulations. 

Lastly, we present measurements of the HTB Qdisc in 
Fig.~\ref{fig:htb-validate}. For this experiment, HTB satisfies the allocation of 
HMM fairness.

\subsection{Experiment 2: Isolating class guarantees}
\label{subsec:exp2}


This experiment illustrate the need for isolating class guarantees, and the inability of the existing link sharing schedulers CBQ and HTB to realize isolation between classes. 
The experiment uses the class hierarchy from Fig.~\ref{fig:hierarchy} (in Sec.~\ref{sec:intro}). In addition to the guarantees shown in the figure, we vary the guarantees of classes $A1, A2, B1, B2$ 
to evaluate three scenarios, labeled as `L', `M', `H', which stands for low, medium, and high differences between the  guarantees. The guarantees 
in the scenarios (in Mbps) are as follows: 

\begin{center}
\begin{tabular}{l | c c c c | c c c}
Class: &  $A1$ & $A2$ & $B1$  & $B2$ &  	$A$ & $B$ & $C$ \\
\hline 
 Low `L' 			& 		140	& 160	& 140	& 160 & 	300	& 300 & 400 \\
 Medium  `M' 	& 		100	& 200	& 100	& 200 & 	300	& 300 & 400\\
 High `H' 		& 		60		& 240	& 60	& 240 & 	300	& 300 & 400
\end{tabular}
\end{center}

\smallskip
The `M' scenario corresponds to the  guarantees shown in Fig.~\ref{fig:hierarchy}. The guarantees of classes $A$, $B$, and $C$ are the same in all three scenarios, and are as shown in Fig.~\ref{fig:hierarchy}.

In the experiment, three leaf classes ($A1, B2, C$) generate traffic.  In the middle of the experiment, in the interval $[10,20]~s$, class $C$ pauses transmissions. 
Since leaf classes $A1$ and $B2$ compete with each other at the level 
of their respective parent classes $A$ and $B$, their allocation 
should be determined by the guarantees of the parents. 
If this is the case, $A1$ and $B2$ receive the same allocation in all three scenarios.

\begin{figure*}[!t]

\centering
	\subfigure[HLS Low (`L') scenario.]{
 	\includegraphics[width=0.31\columnwidth]{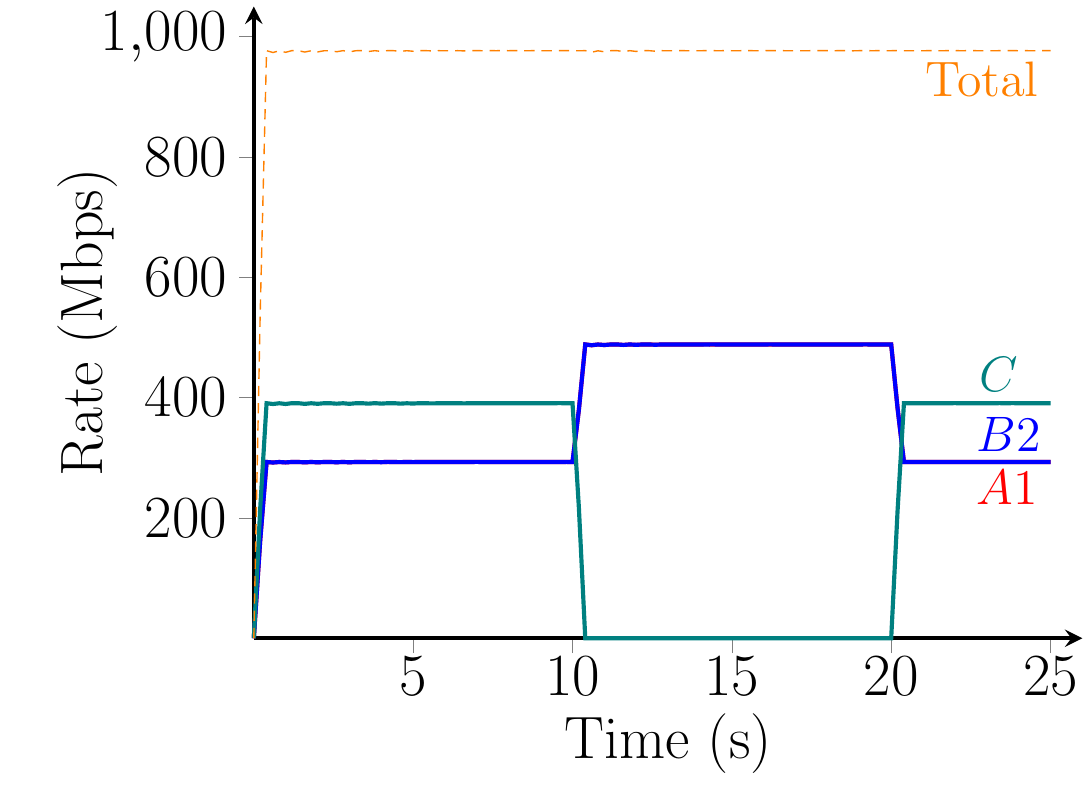}
	\label{fig:hls-unfair-low}
	}
%
	\subfigure[HLS Middle (`M') scenario.]{
 	\includegraphics[width=0.31\columnwidth]{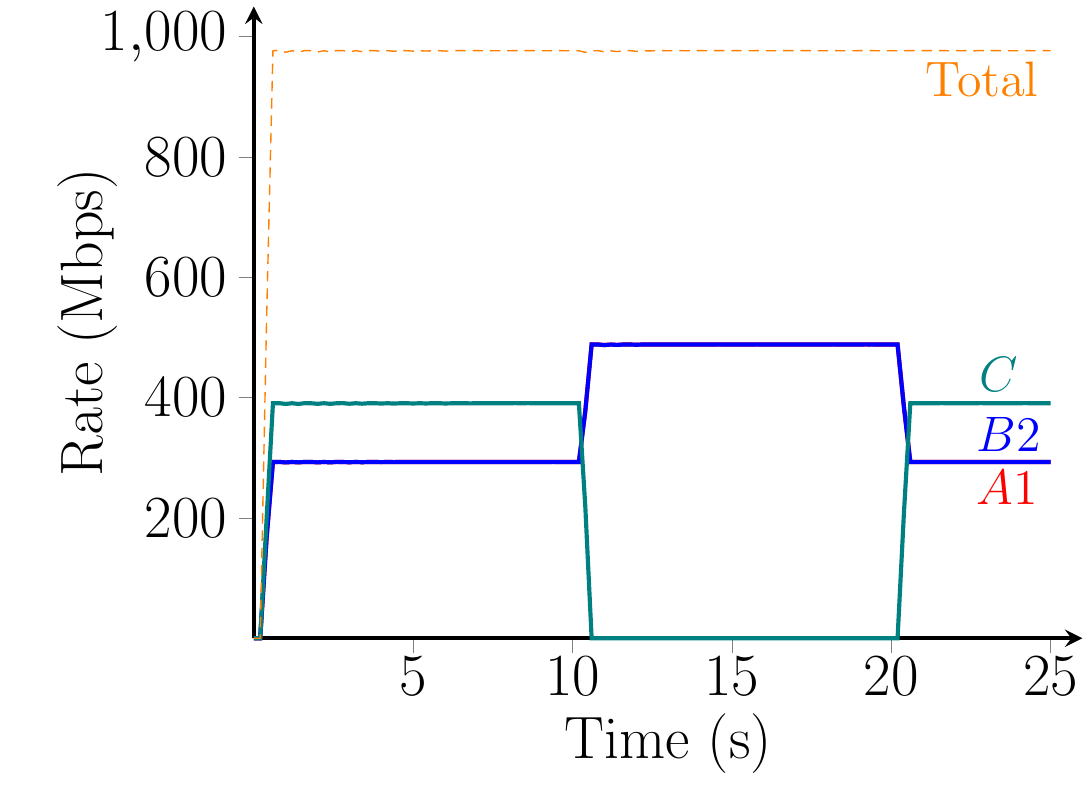}
	\label{fig:hls-unfair-middle}
	}   
%
\subfigure[HLS  High (`H') scenario.]{
 	\includegraphics[width=0.31\columnwidth]{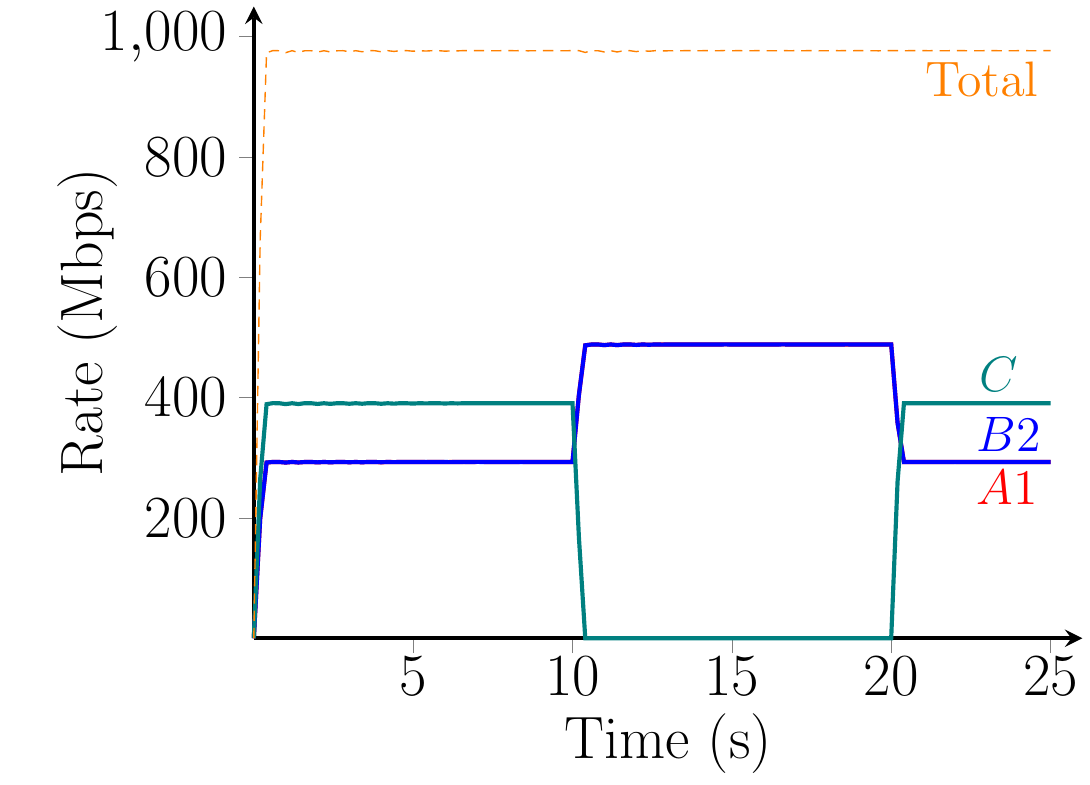}
	\label{fig:hls-unfair-high}
	}   

\centering
	\subfigure[HTB Low (`L') scenario.]{
 	\includegraphics[width=0.31\columnwidth]{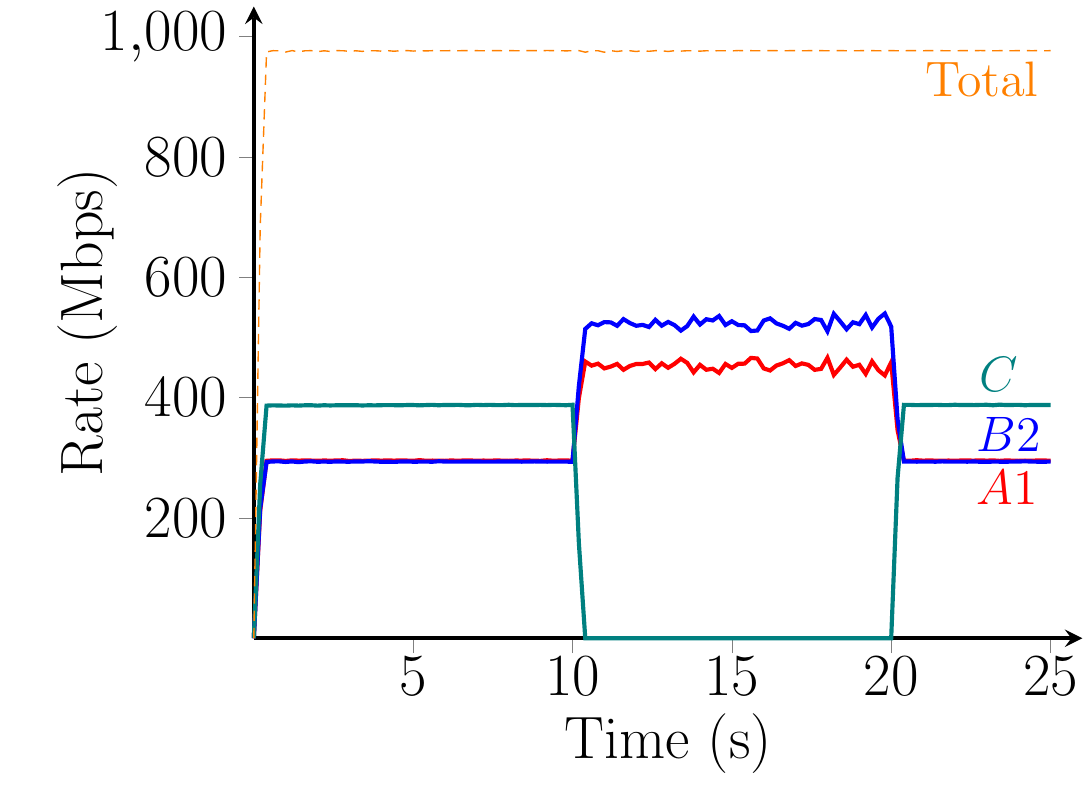}
	\label{fig:htb-unfair-low}
	}
\centering
	\subfigure[HTB Middle (`M') scenario.]{
 	\includegraphics[width=0.31\columnwidth]{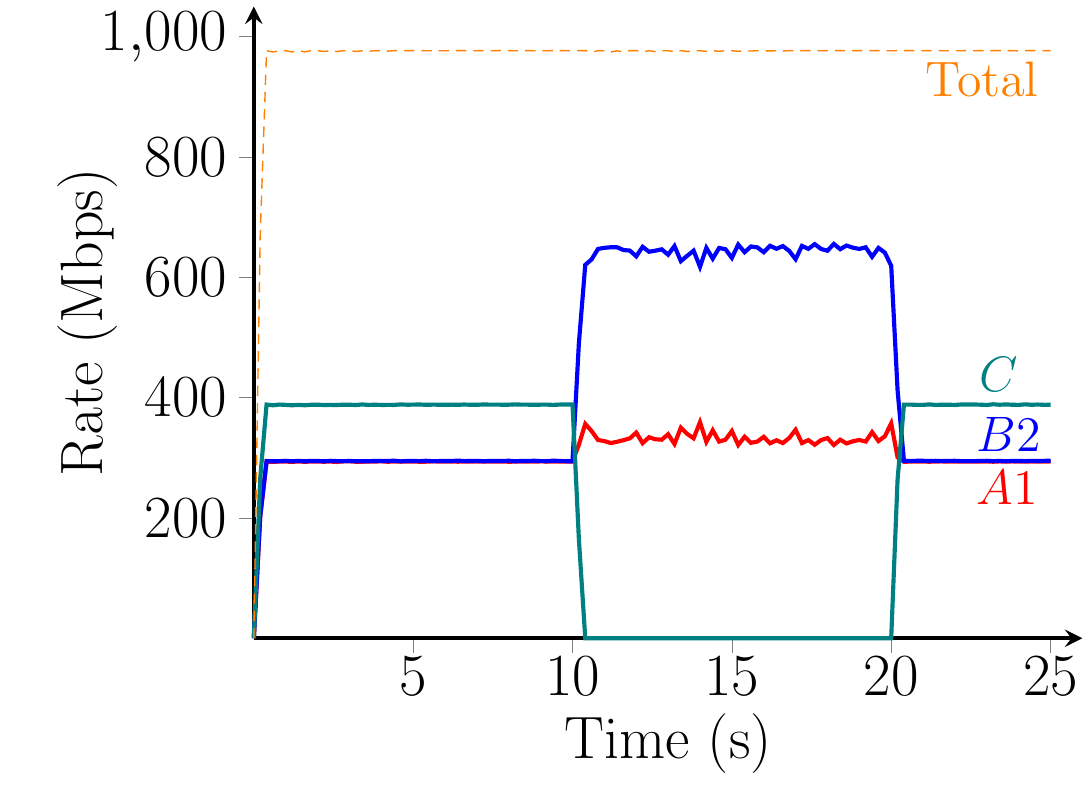}
	\label{fig:htb-unfair-middle}
	}   
%
\subfigure[HTB High (`H') scenario.]{
 	\includegraphics[width=0.31\columnwidth]{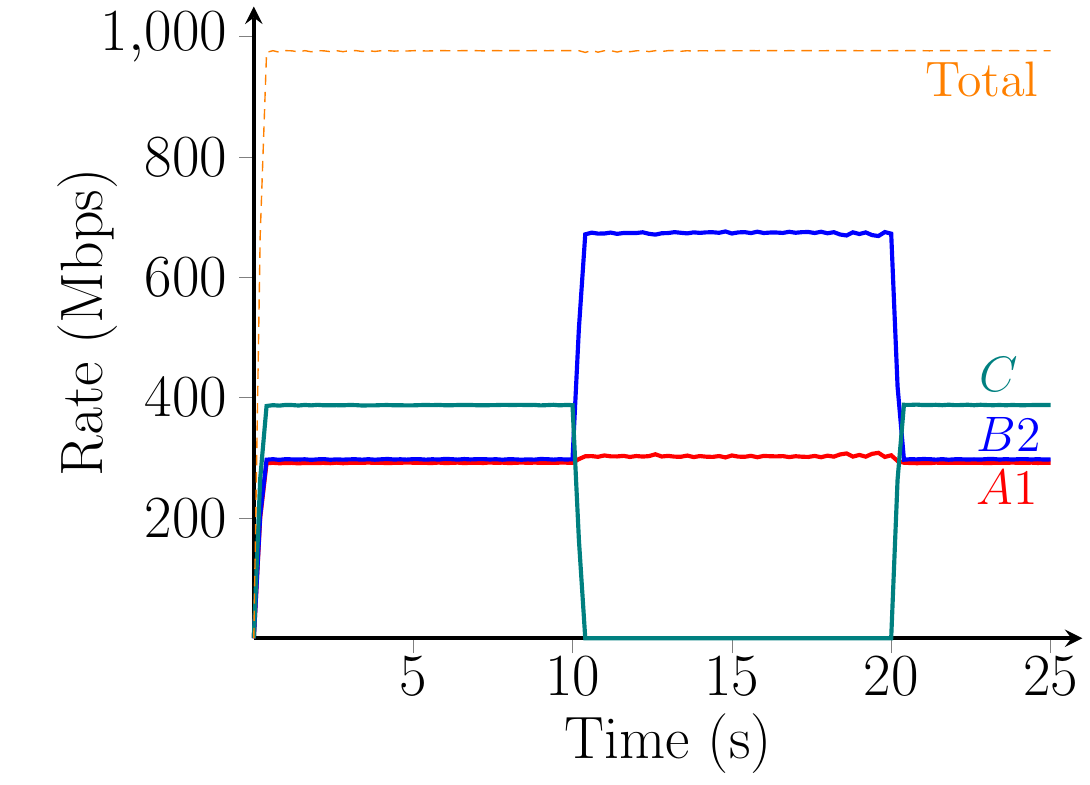}
	\label{fig:htb-unfair-high}
	}   

 \centering
	\subfigure[CBQ  Low (`L') scenario.]{
 	\includegraphics[width=0.31\columnwidth]{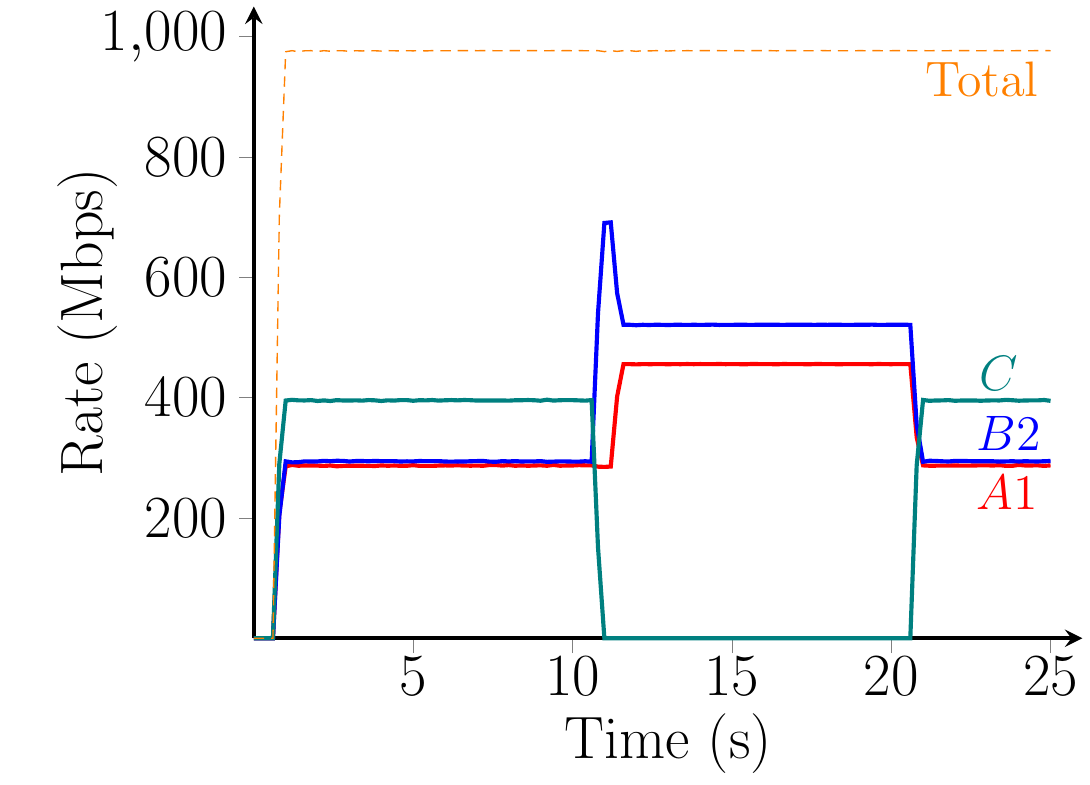}
	\label{fig:cbq-unfair-low}
	}
%
	\subfigure[CBQ Medium (`M') scenario.]{
 	\includegraphics[width=0.31\columnwidth]{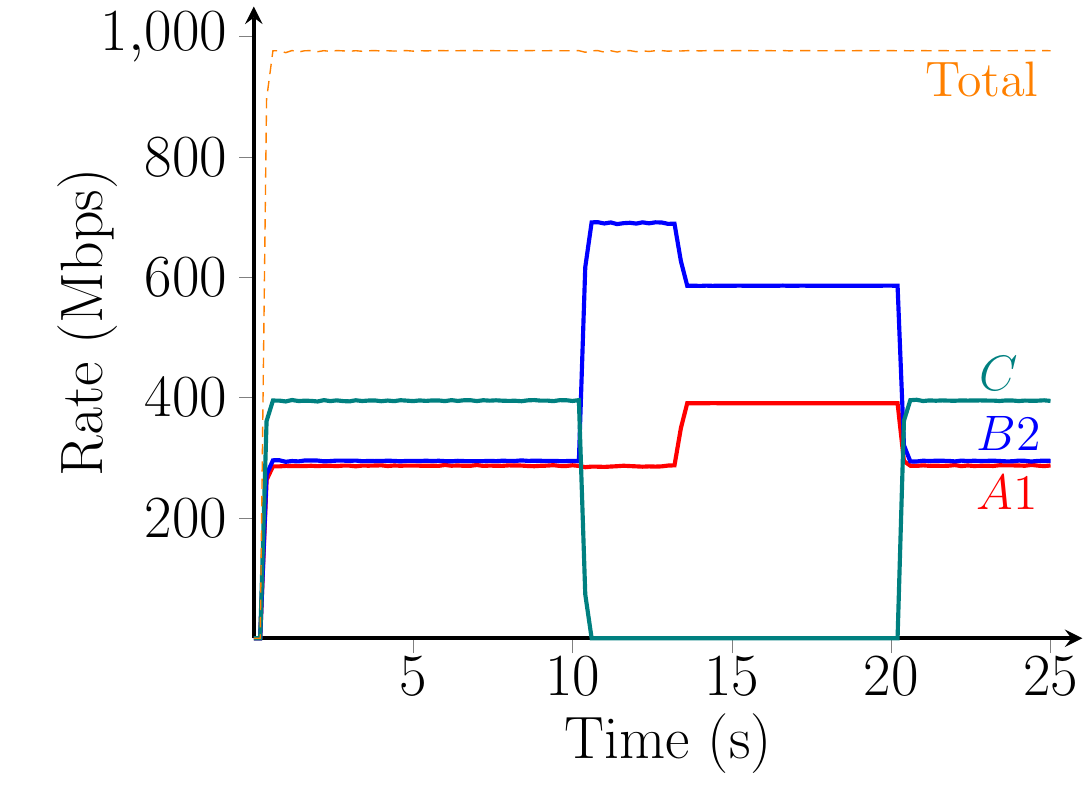}
	\label{fig:cbq-unfair-middle}
	}   
%
\subfigure[CBQ High (`H') scenario.]{
 	\includegraphics[width=0.31\columnwidth]{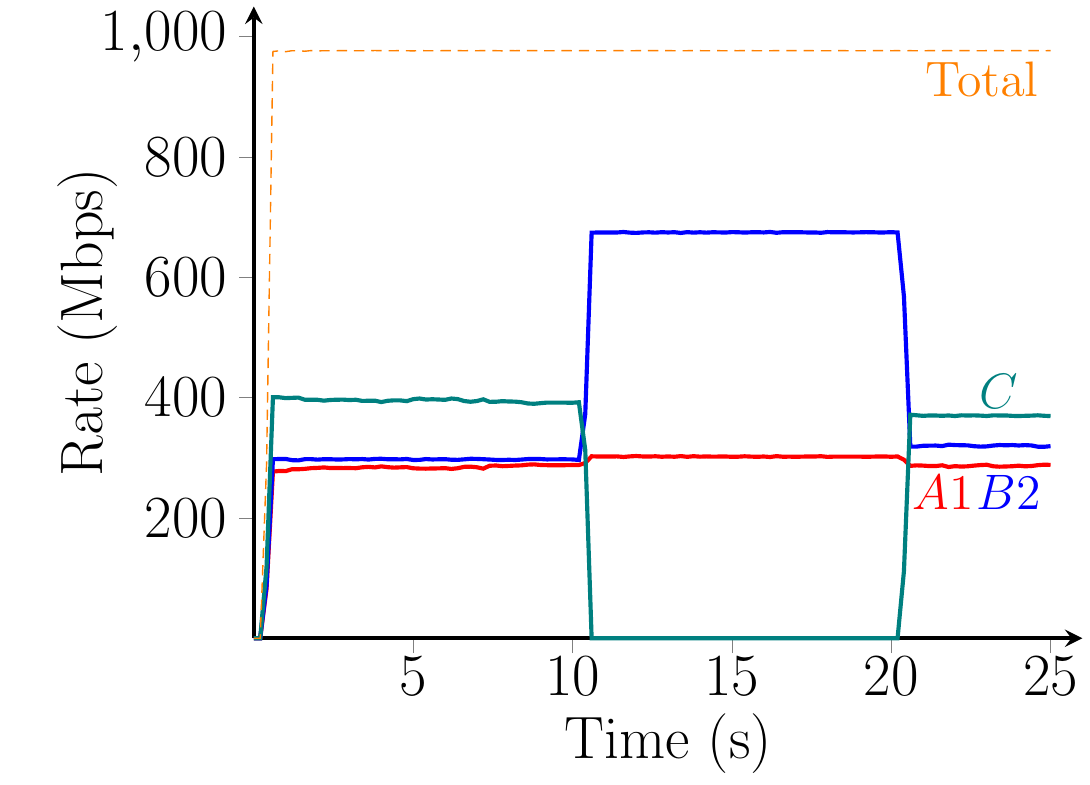}
	\label{fig:cbq-unfair-high}
	}

  \caption{Experiment 2:  {\rm HLS is HMM fair. With HTB and CBQ,  
  class $B$ can increase its allocation by bundling its guarantee and traffic 
  in one child class.}}
 \label{fig:Exp2-unfair}
\end{figure*}

Fig.~\ref{fig:hls-unfair-low}--\ref{fig:hls-unfair-high} show the measured throughput of the HLS Qdisc. In all three scenarios, the throughput of active classes corresponds to the HMM fair allocation. When all three classes are active (in $[0,10]~s$ and $[20,25]~s$) they split the allocation in 
the  ratio $3:3:4$, according to the guarantees of classes~$A, B, C$. 
When class~$C$ drops out, $A1$ and $B2$ split the capacity evenly, since 
$A$ and $B$ have the same guarantee. 

The second row of graphs in Fig.~\ref{fig:Exp2-unfair} presents measurements 
of the HTB Qdisc. 
First note that, for all scenarios the minimum rate guarantees of  internal and leaf classes are maintained at all times. 
When all three classes are active, they have the same allocation as HLS. 
However, when class $C$ drops out, 
classes $A1$ and $B2$ do not split the freed up link capacity evenly. 
Instead, the throughput 
appears to depend on the guarantees  of the active leaf classes $A1$ and $B2$. As seen in Figs.~\ref{fig:htb-unfair-middle} and~\ref{fig:htb-unfair-high}, by increasing the guarantee of class $B2$ and decreasing that of $A1$, the allocation becomes more lopsided.

The throughput in the scenarios under CBQ, depicted in the last row of graphs  
in Fig.~\ref{fig:Exp2-unfair}, shows that CBQ allocates rates in a similar fashion as HTB. That is, 
sharing of bandwidth at the level of internal classes does not respect the guarantees of 
the internal classes. As with HTB, the allocation  appears to be again determined by the guarantees of the leaf classes. 

The observed link sharing of CBQ and HTB indicates a lack of isolation between classes in the hierarchy.  
Here, class~$B$ can manipulate its allocation by bundling its traffic in a single descendant class, at the cost of class $A$.  The HMM fair allocation of HLS does not allow this to happen.  

\newpage 
\begin{center}
\begin{figure}[t]
\centerline{
\begin{tikzpicture}[-triangle 45,sibling distance=10pt, 
  every node/.style = {shape=circle, draw, align=center},level 1/.style={sibling distance=25mm,align=center,},level 2/.style={sibling distance=15mm}]
  \node {root \\ $1000$}
    child { node[fill=gray!20, ] {$A$ \\ 200}} 
    child { node[sibling distance=1em] {$B$ \\ 250} 
        child { node [fill=gray!20, sibling distance=5em]{$B1$ \\ 240} }
        child { node[sibling distance=5em] {$B2$ \\ 10} } 
		child [missing] } 
 child { node {$C$ \\ 250} 
		child [missing]
        child { node [fill=gray!20, sibling distance=5em]{$C1$ \\ 50} }
        child { node[sibling distance=5em] {$C2$ \\ 200} } }
           child { node {$D$ \\ 300}} ; 
\end{tikzpicture}
}
\caption{Class hierarchy in Experiment 3.}
\label{fig:exp3-hierarchy}
\end{figure}
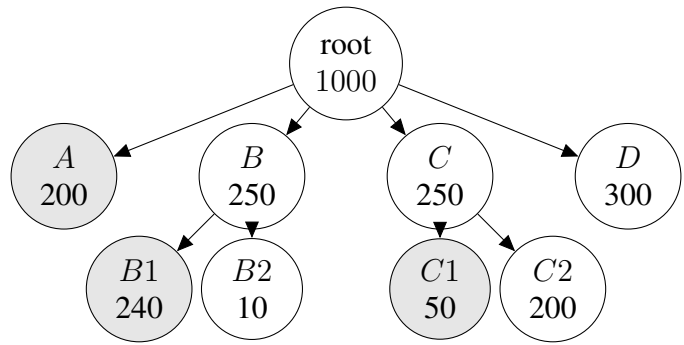
\end{center}

\begin{figure*}[ht]

  \centering
	\subfigure[CBQ qdisc.]{
 	\includegraphics[width=0.31\columnwidth]{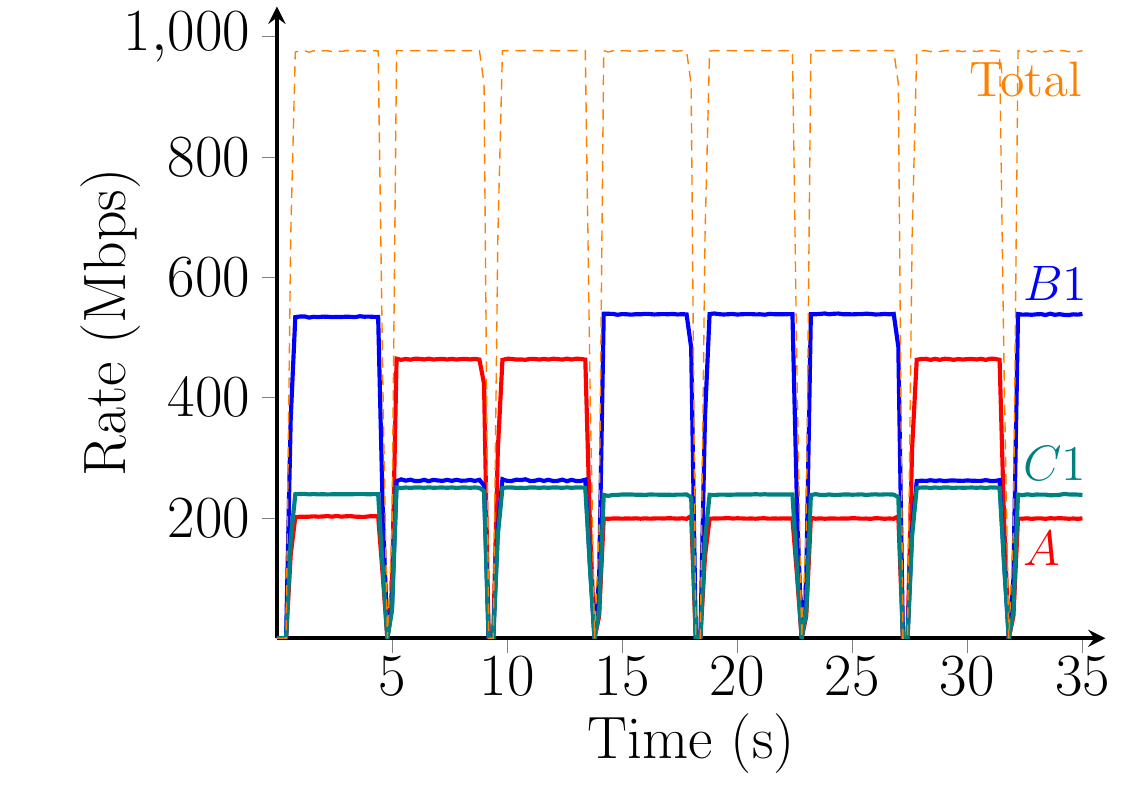}
	\label{fig:cbq-bistabililty}
	}
%
	\subfigure[HTB qdisc.]{
 	\includegraphics[width=0.31\columnwidth]{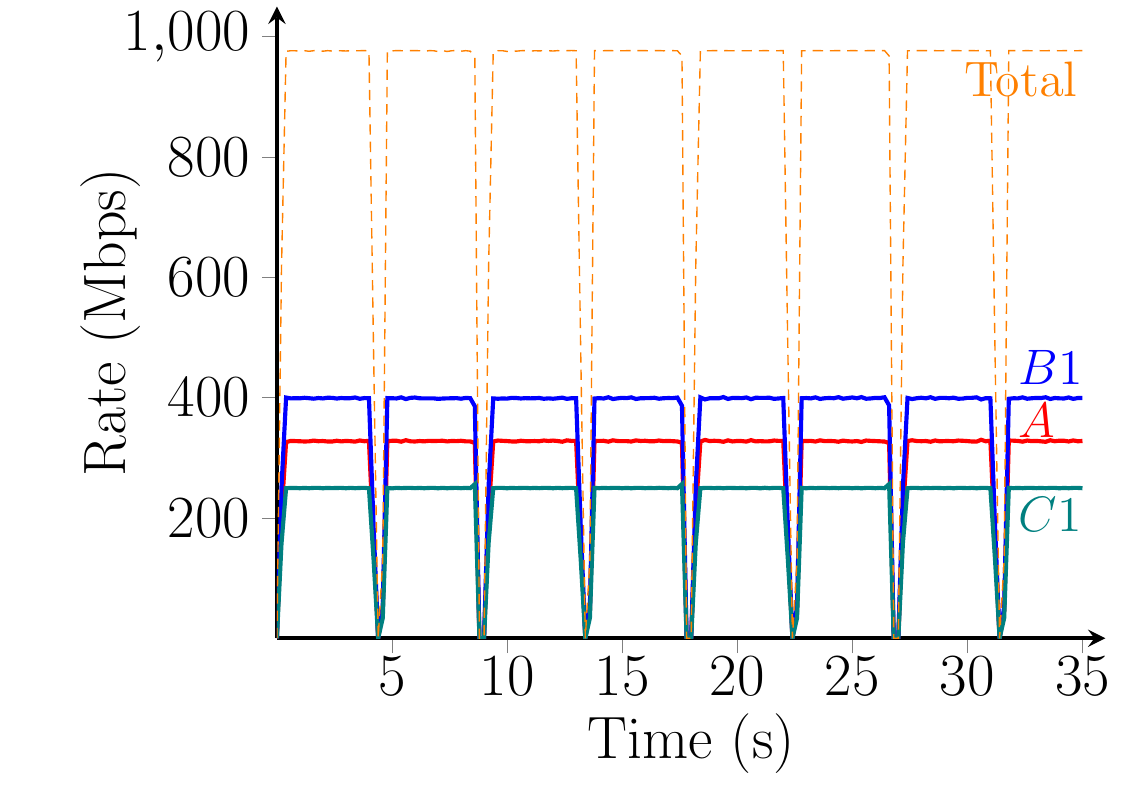}
	\label{fig:htb-bistabililty}
	}   
%
\subfigure[HLS qdisc.]{
 	\includegraphics[width=0.31\columnwidth]{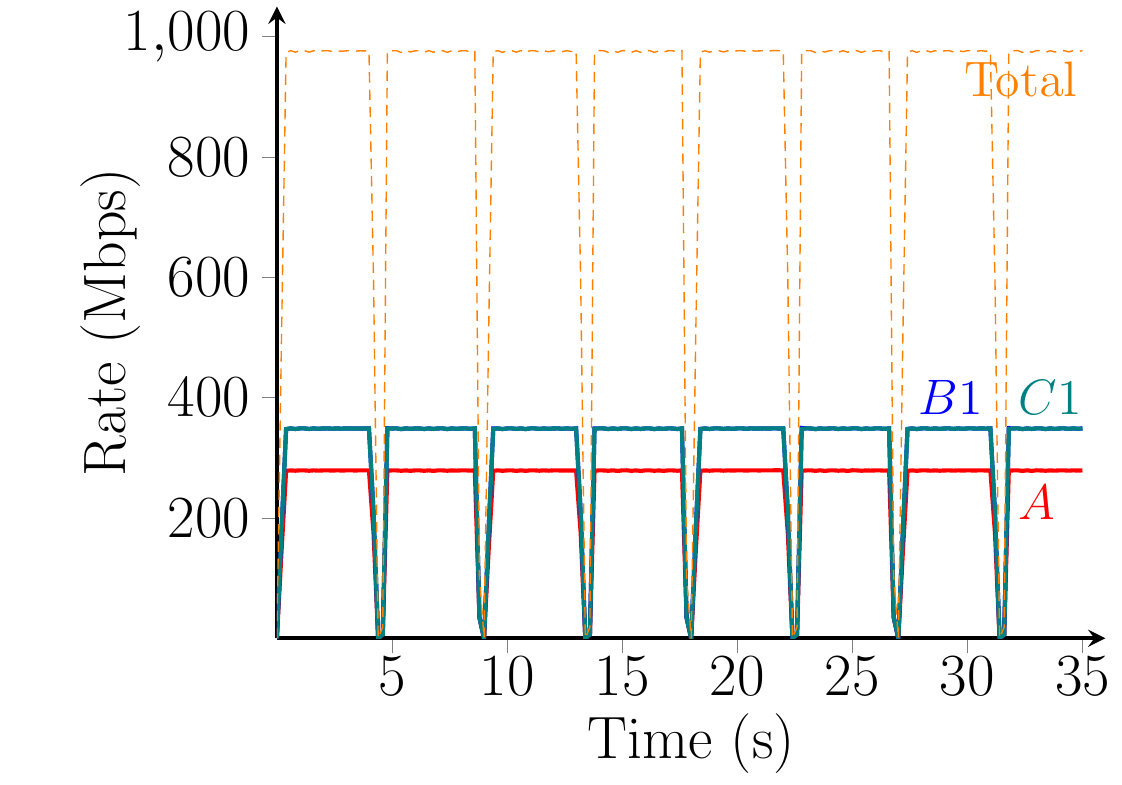}
	\label{fig:hls-bistabililty}
	}   
	
\caption{Experiment 3. Non-determinism.  {\rm With CBQ, stopping and resuming transmissions results in different rate allocations. }}
 \label{fig:Exp3-bistability}
\end{figure*}

\subsection{Experiment 3: Non-determinism}
\label{subsec:exp3}

This experiment shows that,  for a given input, the CBQ qdisc 
does not uniquely determine a rate allocation. 
We consider the class hierarchy  in Fig.~\ref{fig:exp3-hierarchy}, where active classes  ($A, B1, C1$) are indicated by a shaded background. 
In the experiment, classes $A, B1, C1$ repeatedly transmit for $4$~s and then pause for $0.5$~s. This creates an on-off pattern which is repeated several times. 
The transmission pauses are sufficiently long so that the backlog at the scheduler is fully cleared.

Fig.~\ref{fig:cbq-bistabililty} presents the allocations of the CBQ qdisc. 
We note again that the minimum guarantees of leaf and internal classes are never violated. 
The main observation is that when classes resume transmissions after a pause, the link sharing settles to two different allocations. 
In both outcomes, class $C1$ is always kept at $250$~Mbps, the minimum 
guarantee of its parent class. 
In one of the outcomes, the allocation of class~$A$ is close to its minimum 
guarantee, while class~$B1$ consumes all of the remaining capacity. 
In the other outcome, $B1$ receives the minimum guarantee given to its parent, 
and $A$ consumes all other bandwidth. No equitable sharing between classes occurs in either outcome.

We have observed non-deterministic rate allocations in other scenarios, including situations where rate allocation change spontaneously  without pausing traffic sources (e.g., see Fig.~\ref{fig:cbq-unfair-middle} at $\sim \!\!13$ s). 
This non-determinism indicates that the allocation of the 
CBQ qdisc is under-determined, that is, multiple allocations satisfy its link sharing guidelines. We have not observed non-deterministic allocations of CBQ in {\it ns2} simulations of this experiment. 

In Fig.~\ref{fig:htb-bistabililty} and~\ref{fig:hls-bistabililty} we show the outcomes for this experiment  for HTB and HLS, respectively. Both create  unique allocations. 
The results for HTB in Fig.~\ref{fig:htb-bistabililty} show link sharing 
between classes $A$ and $B1$, while class~$C1$, similar to CBQ, is kept at the guarantee of its parent class.  The allocation of HLS in Fig.~\ref{fig:hls-bistabililty} is HMM fair in each `on' phase. Classes $A$, $B1$, and $C1$  
share the link capacity in the ratio 4:5:5, according to the ratio of the guarantees of classes $A$, $B$, and $C$.

\subsection{Experiment 4: Overhead}

We next measure the processing overhead of HLS and compare it to that of CBQ and HTB. Since all schedulers are implemented as  Linux Qdiscs, they share  the performance limitations of the Qdisc framework, in particular, the single Qdisc lock. In order to move the bottleneck of the experimental 
setup to per-packet processing, we replace the 1~Gbps link in Fig.~\ref{fig:exp_topo} between the scheduler and the traffic sink by a 10~Gbps link, and we send small packets. 
We verified that the servers that run the traffic generator and traffic sink are not bottlenecks in 
the experiment. 

The experiment uses the  NetPerf TCP-RR \cite{netperf}, similar to an experiment in \cite[Fig.~6]{Carousel}. TCP senders and receivers send  1-byte packets 
in each direction in a ping-pong fashion. One round of the ping-pong is 
called a transaction. 
The traffic from each TCP sender is mapped to a separate leaf class 
at the scheduler node (in Fig.~\ref{fig:exp_topo}). 
The performance metric is the total number of completed transactions per second. 

We consider two class hierarchies: a perfect binary tree and a flat hierarchy. 
In a binary tree, with $N$~leaf classes, the total number of classes, including the root,  is 
$2N-1$. 
In the flat hierarchy, the root class has~$N$ children which are all leaf classes. 
For HLS we set the weight of every class to one in both scenarios. 
For HTB and CBQ, the rate guarantees are divided evenly between the leaf classes.
In addition to the hierarchical scheduling algorithms HLS, HTB, and CBQ, we also include measurements with   FIFO scheduling. Since FIFO is a classless scheduler, outcomes are not sensitive to the class hierarchy.

Fig.~\ref{fig:overhead-binary} shows the  transactions per second as a function of the number of leaf classes for the binary tree hierarchy. 
All schedulers show roughly the same performance.  
The number of transactions initially increases linearly with the leaf classes 
and saturates at around 300K transactions  (100K $= 10^5$). 
Since HTB limits the number of levels in the class hierarchy, the binary tree 
hierarchy cannot be increased beyond 128 leaf classes.
The fact that FIFO sometimes has worse results than the hierarchical 
schedulers  indicates the degree of randomness in 
experiments that involve a large number of TCP  flows.

Fig.~\ref{fig:overhead-flat} depicts the results for the flat hierarchy. Here, the number of transactions initially increases and plateau at around~$300K$ transactions, similar to binary tree scenario.
After around 256 leaf classes, however, the performance of the classful schedulers declines. A comparison with FIFO points to a  performance bottleneck that arises during  packet classification, which impacts all classful schedulers in the same way.
The classification compares the packet destination port to the port associated to each leaf class, and the number of comparisons grows linearly with the number of leaf classes.

\begin{figure}[t]

\centering
	\subfigure[Binary tree class hierarchy.]{
 	\includegraphics[width=0.48\columnwidth]{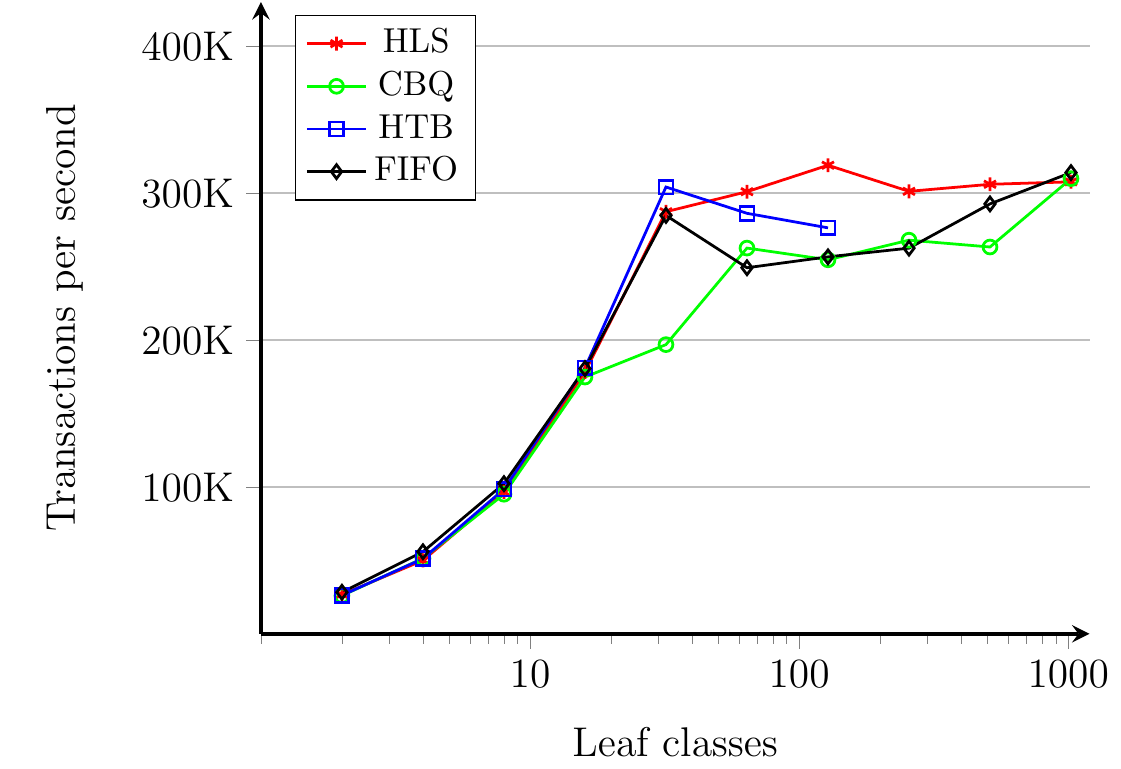}
	\label{fig:overhead-binary}
	}   
%
	\subfigure[Flat class hierarchy.]{
 	\includegraphics[width=0.48\columnwidth]{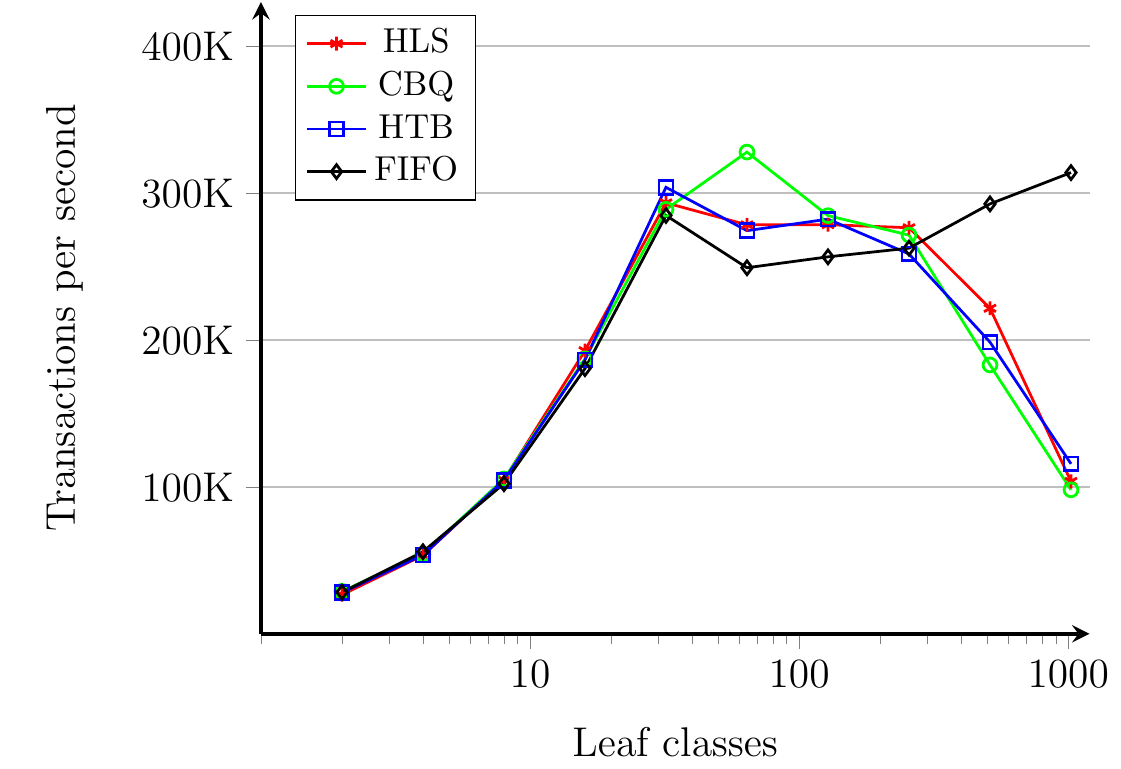}
	\label{fig:overhead-flat}
	}   

\caption{Experiment 4. Overhead.
}
\label{fig:Exp4-overhead}


\end{figure}
The experiment leads us to conclude that the HLS Qdisc does not 
incur a performance penalty, compared to HTB and CBQ. In fact, 
since the schedulers perform similarly to FIFO, none of the 
hierarchical schedulers presents a  bottleneck. 

We emphasize that the outcome of this  experiment is   
sensitive to the configuration of filters that map 
packets to traffic classes. In Fig.~\ref{fig:Exp4-overhead}, 
the mapping of packets to classes is done progressively. Each internal node in 
the hierarchy has (two) filter expressions for mapping 
traffic to its child classes. Alternatively, the mapping 
can be performed at the root Qdisc for all leaf classes. If this is done, 
the results show a precipitous drop 
of completed transactions when the number of leaf classes exceeds 100.

\section{Related Work}
\label{sec:related}

There are several reasons for the recent surge of interest in shaping and scheduling 
algorithms.
First, the increased flexibility of recent programmable packet switches 
has  enabled customization of scheduling algorithms to application requirements 
\cite{P4,PIFO,Siva20}. Second, the Ultra-Reliable Low-Latency Communication service category in 5G networks, which  guarantees latencies below 1~ms has led to 
standardization efforts by the IEEE (for Layer-2) and by the IETF (for Layer-3) 
for compatible protocol frameworks \cite{tsn} and traffic control algorithms 
\cite{Leboudec18}. Third, an increased demand for fine-grain control of traffic 
in data centers has created a need for advanced packet scheduling methods 
at servers~\cite{B4,datacenter-tc}. 

These efforts benefit from an intense period of 
research in the 1990s that created many of the 
scheduling and shaping  algorithms in use today 
\cite{DRR,WFQ-keshav,ZhangProc,CBQ}. 
Recent research on packet scheduling has put emphasis on 
generality, e.g., PIFO~\cite{PIFO}, UPS~\cite{UPS}, and 
efficient implementations, for example, Carousel~\cite{Carousel}, Eiffel~\cite{Eiffel}, Loom~\cite{Loom}, CQ~\cite{sharma20}. 

Most relevant to our paper is the claim in \cite{PIFO} of 
realizing HPFQ by a hierarchy of PIFO queues. 
However, the claim holds only when packets have a fixed size. 
For variable-sized packets and classes with different weights the arrival of a packet may require changing the relative order of packets in the PIFO buffers. By design, PIFO does not support reshuffling buffered packets.   

BwE \cite{BwE} performs a centralized rate allocation 
for hierarchically organized inter-data center traffic, 
which computes end-to-end max-min fair rate allocations for a 
network setting, which are enforced by HTB ceiling rates. 
Interestingly, in \cite[Sec.9]{BwE} it is argued that fair queueing  
is not suitable since `weights are insufficient for 
delivering user guarantees.'  By showing the equivalence 
of rate guarantees and weights  in Sec.~\ref{sec:hierarchy}, our paper shows that the above statement requires a correction. 

We have not included HFSC \cite{HFSC2,HFSC-techrep} in this paper, 
even though it is another hierarchical scheduler available in Linux. 
HFSC is hybrid scheduler that deterministically guarantees service 
curves and shares   excess bandwidth with fairness objectives. 
As pointed out in \cite{HFSC2} it is, in general, not possible to simultaneously guarantee  the service curves of HFSC and its fairness criteria. 
Since HFSC resolves conflicting  guarantees by giving priorities to service curves, the role of link sharing is limited. 
We also note that HFSC realizes so-called `lower  service curves' 
\cite{Book-LeBoudec}. Rate guarantees  of these service curves share 
a drawback with the VirtualClock scheduler~\cite{VC}, where 
a class that is served above its guaranteed rate for some time, 
may be later served at a rate below its guarantee.
Differently, fair scheduling algorithms  realize 
`strict service curves'~\cite{Book-LeBoudec}, which ensure 
rate guarantees for every time interval where a class is backlogged. 


\section{Conclusions}
\label{sec:concl}

We presented a  round-robin scheduler for hierarchical link sharing 
 that ensures rate guarantees and isolation between classes, and which is suitable for supporting high line rates. 
The presented HLS  scheduler resolves shortcomings of  deployed 
hierarchical link sharing algorithms when distributing excess capacity to 
traffic classes. The link sharing in HLS is strategy-proof in that a class that needs more 
bandwidth cannot increase its allocation by increasing its transmissions 
or misrepresenting the class hierarchy of its descendants.  
We have shown that the implementation of HLS does not create a 
performance bottleneck. 
In future work, we will extend the Qdisc implementation of HLS to also 
enforce maximum (ceiling) rates.


\newpage 
\appendix 
\section{Supplemental information on CBQ and HTB implementations}

Since available documentation for Linux schedulers often focuses on 
configuration issues, we provide conceptual descriptions of HTB and CBQ. 
First, we briefly discuss two building blocks: 

\smallskip
\noindent 
\underline{DRR scheduler.}
The scheduler maintains one FIFO queue for each class. Queues with a backlog are visited in a round-robin fashion. Each class has a  counter 
for the number of bytes that the class can transmit. When 
a packet is transmitted, the counter is decremented by the packet size. In a visit by the round robin, a class 
can transmit as long as it has packets and the counter does not become negative. At the start of a new round of the round robin, 
the counters are incremented by the quantum which is set proportional to the class guarantee. 

\smallskip
\noindent 
\underline{Token Bucket.} A token bucket has two parameters, a bucket size $b$ and a rate $r$. 
The bucket content is a counter, which is initialized to $b$. When a packet of size $L$ is transmitted, $L$ tokens 
are removed from the 
bucket. If there are less than~$L$ tokens in the bucket, the packet has to wait until $L$ tokens 
are available. Tokens are added at a constant rate of $r$, but the bucket content may not exceed $b$.

\subsection*{CBQ}
The following description of CBQ is based on the Linux qdisc code \cite{CBQ-code} 
and  a reference implementation of CBQ in the {\it ns2} simulator~\cite{ns2}. 
The implementations sometimes deviate from the 
description in \cite[Appendix A]{CBQ}. 

\begin{figure}[!h]
    \centering
    \includegraphics[width=0.6\columnwidth]{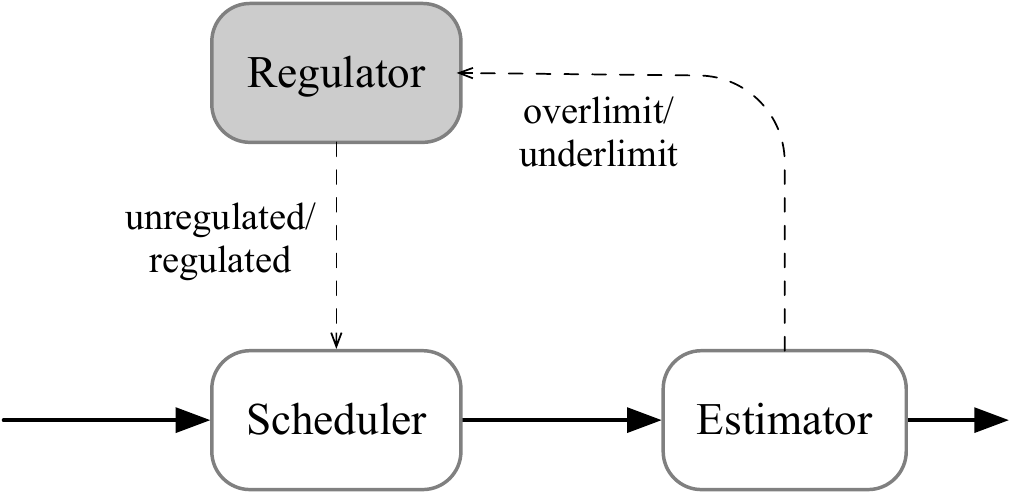}
\caption{Components of CBQ Scheduler.}
    \label{fig:cbq}
\end{figure}

The main components of a CBQ scheduler, shown in Fig.~\ref{fig:exp_topo} are (1) an {\it estimator}, (2) a {\it regulator}, and (3) a {\it scheduler}. 
The estimator  computes the transmission rate of each class using an exponential moving average 
and determines if a class is serviced below or above its guaranteed rate.  The class is {\it underlimit} in the former case and {\it overlimit} in the latter. 
The {\it regulator} uses the output of the  estimator to tag classes as {\it regulated} or {\it unregulated}. Regulated classes are not considered for transmission. This is different from \cite[Appendix A]{CBQ}, which describes CBQ as having two schedulers, one for regulated and for unregulated classes. The scheduler runs Weighted Round Robin, which is  essentially identical to DRR, to transmit packets from unregulated classes, where the quantum of each class is determined from the class guarantee. 
The regulator inspects the underlimit/overlimit state in the class hierarchy. A class is unregulated (1) if it is  underlimit or (2) if it has an ancestor that is underlimit and that has no underlimit descendants with a backlog. 
The first rule ensures that a class receives its rate guarantee. The second rule determines when a class can borrow bandwidth in excess of its guarantee. In \cite{CBQ} these rules are referred to as {\it formal link sharing} guidelines. 

To reduce the complexity involved in applying the  
formal link sharing guidelines, 
CBQ proposes two approximations, which are referred to as  ancestor-only and top-level link sharing guideline. The Linux implementation of CBQ only supports 
top-level link sharing. 

\subsection*{HTB}

Since HTB supports both minimum class guarantees and upper bounds on the rate of a class, it is both a 
scheduler and a shaper. Possibly because of its name,  discussions 
of HTB often focus on its shaping operation, and it is frequently 
characterized as a pure traffic shaper. 
The following description, which is based on the HTB qdisc code \cite{HTB-code}, is simplified in parts. 

In HTB, each class has two token buckets, an assured bucket and a ceiling bucket. The assured and ceiling buckets of class~$i$ are filled at rate $AR_i$ and $CR_i$, respectively, with $CR_i \ge AR_i$. 
For an internal class~$i$, reasonable choices should satisfy 
\[
AR_i \ge \sum_{j \in {\rm child}(i)} AR_j 
\quad \text{and} \quad 
CR_i = \max_{j \in {\rm child}(i)} CR_j  \, . 
\]
In the experiments we set $CR_i$ to the link capacity $C$ for all classes, 
which effectively disables traffic shaping. 
For  transmission of a packet with size $L$ from a class, $L$ tokens are removed from both 
the assured and  ceiling buckets, as well as from the  buckets of the ancestors of the class. 

The filling level of the token buckets
determine the state of a class, where each state is associated with a color:

\begin{center}
\begin{tabular}{ l | l | l }
& State & Color  \\ 
\hline
Assured bucket has tokens & Can send & green \\
Assured bucket is empty, but ceiling bucket has tokens & Can borrow & yellow \\
Ceiling bucket is empty & Can't borrow & red  
\end{tabular}
\end{center}
A green class can always transmit and a red class cannot transmit. A yellow class can transmit by borrowing tokens if its parent is green.  If the parent is  also yellow, the class tries to borrow from the next ancestor. This continues until an ancestor is reached that is either green or red. 
If a red ancestor is reached, the class cannot transmit. 
If a green ancestor is reached, the class can transmit, 
and tokens are removed from 
the ceiling buckets of the class and its ancestors, as well as 
from the assured buckets of green ancestors. 

\begin{figure}[t]
\centerline{
\begin{tikzpicture}[-triangle 45,font=\small,level distance=50pt,
  every node/.style = {shape=circle, draw, align=center, minimum size=9mm},level 1/.style={sibling distance=30mm,align=center,},level 2/.style={sibling distance=15mm}]
  \node[fill=teal!50]{}
       child { node[fill=yellow,sibling distance=1em] {} 
              child { node [fill=teal!50, sibling distance=5em]{$\bm{1}$}}  
              child { node [fill=yellow, sibling distance=5em]{} 
              	child { node [fill=yellow, sibling distance=5em]{$\bm{2}$}}
         			child { node [fill=red!50, sibling distance=5em]{$\bm{3}$} }}
              child [missing]
              }
 	child { node [fill=red!50] {} 
        child { node [fill=teal!50, sibling distance=5em]{$\bm{4}$}}
         child { node [fill=yellow, sibling distance=5em]{$\bm{5}$} }}; 
\end{tikzpicture}
}
\caption{States of classes in HTB.}
\label{fig:HTB-states}
\end{figure}
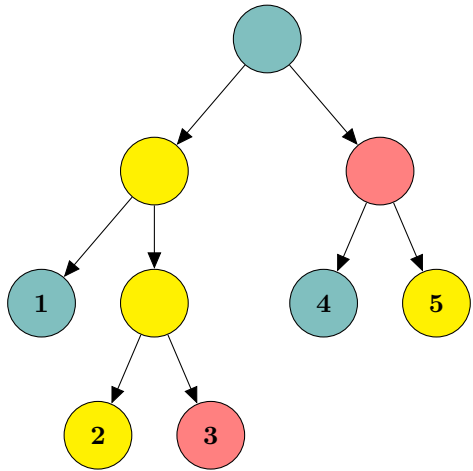
Consider Fig.~\ref{fig:HTB-states} as an example, which depicts a class hierarchy where the states of the classes are indicated by the shading of the circles. 
Among the leaf classes, classes $1$ and $4$ can transmit because they are green, and class $3$ cannot because it is red. Class $2$ is yellow, and goes out to borrow bandwidth. It cannot borrow from its parent or grandparent, because they are both yellow. However, it can borrow from the root, 
its great-grandparent, and is allowed to transmit. (In a workconserving scheduler, the root class is always green.) Class $5$, which is also yellow, has a red parent and is therefore blocked from transmission.

HTB organizes green classes into groups and prioritizes the groups. (HTB also permits the  configuration of class priorities, which we do not discuss here.) 
The grouping of green classes is based on their position  in the class hierarchy. A leaf class~$i$ is assigned ${\rm level}(i)=0$, and an internal class~$i$ is 
assigned   
\[
{\rm level}(i) = 1+ \max_{j \in {\rm child}(i)} {\rm level}(j) \, . 
\]
Classes with the same level are placed in the same group, and the group with the lowest level becomes the top group. 
HTB maintains one DRR scheduler for each group, where 
only the DRR scheduler  of the top group is active.  
In any of the DRR schedulers, the quantum  of a class is set 
proportional to its assured rate.

As long as there are green leaf classes, they are served in a DRR fashion. If there is no green leaf class, then the DRR scheduler for the green classes at level~1 becomes active. This  scheduler performs a DRR over the yellow leaf classes with a green parent.  If there is no green class at level~1, then the DRR scheduler for the green classes at level~2 is activated. This DRR scheduler serves yellow leaf classes with a yellow parent and a green grandparent. And so on. 
After each transmission, all classes are checked for a change of their color, which may also cause a change of the top group.

\end{document}